\documentclass[a4paper,11pt]{article}

\usepackage{microtype}
\usepackage{graphicx}

\graphicspath{{./Figures/}}

\usepackage{fullpage}

\author{Shay Golan\\ Bar Ilan University\\\texttt{golansh1@cs.biu.ac.il}
\and Tsvi Kopelowitz\\ University of Waterloo\\ \texttt{kopelot@gmail.com}
\and Ely Porat\\ Bar Ilan University\\\texttt{porately@cs.biu.ac.il}}

\usepackage{makeidx}
\usepackage{amsfonts}
\usepackage{amssymb}
\usepackage{amsthm}
\usepackage{amsmath}
\usepackage{tabularx, booktabs}
\usepackage{tabu}
\usepackage{colortbl}
\usepackage{enumitem}
\usepackage{subfigure}
\usepackage[table]{xcolor}
\usepackage[hidelinks]{hyperref}

\newtheorem{theorem}{Theorem}
\newtheorem{lemma}[theorem]{Lemma}

\newtheorem{definition}[theorem]{Definition}

\usepackage{enumitem}

\usepackage{listings}

\newtheorem{claim}{Claim}

\usepackage{marginnote}
\usepackage{clrscode3e}
\def\polylog{\operatorname{polylog}}
\newcommand{\ceil}[1]{\left\lceil{#1}\right\rceil}
\newcommand{\floor}[1]{\left\lfloor{#1}\right\rfloor}

\newcommand{\modulo}{\operatorname{mod}}
\newcommand{\naive}{na\"{\i}ve }		
\newcommand{\si}{SI}	
\newcommand{\SI}{SI}

\newcommand{\Enqueue}{\normalfont{\texttt{Enqueue}}}
\newcommand{\Dequeue}{\normalfont{\texttt{Dequeue}}}

\newcommand{\ti}[2]{\texttt{text\_interval}({#1},{#2})}
\newcommand{\candidates}[2]{\mathcal C({#1},{#2})}
\newcommand{\candidatesUi}[2]{\mathcal C_{ap}({#1},{#2})}
\newcommand{\candidatesNonUi}[2]{\overline{\mathcal C_{ap}}({#1},{#2})}

\newcommand{\candidatesNonUix}[2]{\overline{\mathcal C_{ap}^{\textit{left}}}({#1},{#2})}
\newcommand{\candidatesNonUiy}{\overline{\mathcal C_{ap}^{\textit{right}}}(I,\alpha)}

\newcommand{\lv}{\textit{L}}
\newcommand{\rv}{\textit{R}}

\title{Streaming Pattern Matching with $d$ Wildcards\footnote{Part of this work took place while the second author was at University of Michigan. This work is supported in part by the Canada Research Chair for Algorithm Design, NSF grants CCF-1217338, CNS-1318294, and CCF-1514383, by ISF grant 1278/16, and by the BIU Center for Research in Applied
Cryptography and Cyber Security in conjunction with the Israel
National Cyber Bureau in the Prime Minister's Office .}}

\begin{document}

\maketitle

\begin{abstract}
In the pattern matching with $d$ wildcards problem one is given a text $T$ of length $n$ and a pattern $P$ of length $m$ that contains $d$ wildcard characters, each denoted by a special symbol $'?'$. A wildcard character matches any other character. The goal is to establish for each $m$-length substring of $T$ whether it matches $P$.
In the streaming model variant of the pattern matching with $d$ wildcards problem the text $T$ arrives one character at a time and the goal is to report, before the next character arrives, if the last $m$ characters match $P$ while using only $o(m)$ words of space.

In this paper we introduce two new algorithms for the $d$ wildcard pattern matching problem in the streaming model.
The first is a randomized Monte Carlo algorithm that is parameterized by a constant $0\leq \delta \leq 1$. This algorithm uses $\tilde{O}(d^{1-\delta})$ amortized time per character and $\tilde{O}(d^{1+\delta})$ words of space.
The second algorithm, which is used as a black box in the first algorithm, is a randomized Monte Carlo algorithm which uses $O(d+\log m)$ worst-case time per character and $O(d\log m)$ words of space.

\end{abstract}
\thispagestyle{empty}
\setcounter{page}{0}.
\newpage

\section{Introduction}	
We investigate the \emph{pattern matching with $d$ wildcards problem} (PMDW)  in the \emph{streaming model}.
Let $\Sigma$ be an alphabet and let $'?'\notin\Sigma$ be a special character called the \emph{wildcard character} which matches any character in $\Sigma$.
The PMDW problem is defined as follows.
Given a \emph{text} string $T=t_0t_1\dots t_{n-1}$ over $\Sigma$ and a \emph{pattern} string $P=p_0p_1\dots p_{m-1}$ over alphabet $\Sigma\cup \{?\}$ such that $P$ contains exactly $d$ wildcard characters, report all of the occurrences of $P$ in $T$. This definition of a match is one of the most well studied problems in pattern matching~\cite{FP74,MR92,Indyk98,Kalai02,CH02,CC07}.

\paragraph{The streaming model.} The advances in technology over the last decade and the massive amount of data passing through the internet has intrigued and challenged computer scientists, as the old models of computation used before this era are now less relevant or too slow. To this end, new computational models have been suggested to allow computer scientists to tackle these technological advances. One prime example of such a model is the \emph{streaming} model~\cite{AMS99,HRR99,Muthukrishnan05,KNPW11}. Pattern matching problems in the streaming model are allowed to preprocess $P$ into a data structure that uses space that is sublinear in $m$ (notice that space usage during the preprocessing phase itself is not restricted). Then, the text $T$ is given online, one character at a time, and the goal is to report, for every integer $\alpha\ge m-1$, whether $t_{\alpha-m+1}\dots t_\alpha$ matches $P$. This reporting must take place before $t_{\alpha+1}$ arrives. Throughout this paper we let $\alpha$ denote the index of the last text character that has arrived.

Following the breakthrough result of Porat and Porat~\cite{PP09}, recently there has  been a rising interest in solving pattern matching problems in the streaming model~\cite{BG14,EJS10,LLZ12,RGM13,JPS13,CFPSS15B,CFPSS16}. However, this is the first paper to directly consider the important wildcard variant.

\paragraph{Related work.}
Notice that one way for solving PMDW (not necessarily in the streaming model), is to treat $'?'$ as a regular character, and then run an algorithm that finds all occurrences of $P$ (that does not contain any wildcards) in $T$ with up to $k=d$ mismatches. This is known as the $k$-mismatch problem~\cite{LV86,PL07,ALP04,CEPR10,CEPR09,CP07,CFPSS16}.
The most recent result by Clifford et al.~\cite{CFPSS16} for the $k$-mismatch problem in the streaming model implies a solution for PMDW in the streaming model that uses $O(d^2 \polylog m)$ words\footnote{We assume the RAM model where each word has size of $O(\log n)$ bits.} of space and $O(\sqrt{d}\log d + \polylog m)$ time per character. Notice that Clifford et al.~\cite{CFPSS16} focused on solving the more general $k$-mismatch problem.

We mention that while our work is in the streaming model, in the closely related online model (see~\cite{CS10,CJPS13}), which is the same as the streaming model without the constraint of using sublinear space, 	
Clifford et al.~\cite{CEPP11} presented an algorithm, known as the \emph{black box} algorithm, which solves several pattern matching problems. When applied to PMDW, the black box  algorithm uses $O(m)$ words of space and $O(\log^2 m)$ time per arriving text character.
In the offline model the most efficient algorithms for PMDW take  $O(n\log m)$ time and were introduced by Cole and Hariharan~\cite{CH02} and by Clifford and Clifford~\cite{CC07}.

\subsection{New results}
We improve upon the work of Clifford et al.~\cite{CFPSS16}, for the special case that applies to PMDW,
by introducing the following algorithms (the $\tilde{O}$ notation hides logarithmic factors). Notice that Theorem~\ref{thm:tradeoff} improves upon the results of Clifford et al.~\cite{CFPSS16} whenever $\delta > 1/2$. We also emphasize that our proof of Theorem~\ref{thm:tradeoff} makes use of Theorem~\ref{thm:main}.

\begin{theorem} \label{thm:main}
There exists a randomized Monte Carlo algorithm for the PMDW problem in the streaming model that succeeds with probability $1-1/poly(n)$, uses $O(d\log m)$ words of space and spends $O(d+\log m)$ time per arriving text character.
\end{theorem}

\begin{theorem} \label{thm:tradeoff}
For any constant $0\leq \delta \leq 1$
there exists a randomized Monte Carlo algorithm for the PMDW problem in the streaming model that succeeds with probability $1-1/poly(n)$, uses $\tilde{O}(d^{1+\delta})$ words of space and spends $\tilde{O}(d^{1-\delta})$ amortized time per arriving text character.
\end{theorem}

\subsection{Algorithmic Overview}
Our algorithms make use of the notion of a \emph{candidate}, which is a location in the last $m$ indices of the current text that is currently considered as a possible occurrence of $P$. As more characters arrive, it becomes clear if this candidate is an actual occurrence or not. In general, an index continues to be a candidate until the algorithm encounters proof that the candidate is not a valid occurrence (or until it is reported as a match). The algorithm of Theorem~\ref{thm:main} works by obtaining such proofs efficiently.

\paragraph{Overview of algorithm for Theorem~\ref{thm:main}.}
For the streaming pattern matching problem without wildcards, the algorithms of Porat and Porat~\cite{PP09} and Breslauer and Galil~\cite{BG14} have three major components\footnote{The algorithms of Porat and Porat~\cite{PP09} and Breslauer and Galil~\cite{BG14} are not presented in this way. However, we find that this new way of presenting our algorithm (and theirs) does a better job of explaining what is going on.}. The first component is a partitioning of the interval $[0,m-1]$  into \emph{pattern intervals} of exponentially increasing lengths.
Each pattern interval $[i,j]$ corresponds to a \emph{text interval} $[\alpha-j+1,\alpha-i+1]$, where $\alpha$ is the index of the last text character that arrived\footnote{The first pattern interval starts at $0$, and so the last text interval ends at location $\alpha+1$, which is a location of a text character that has yet to arrive. To understand why this convention is appropriate, notice that initially every text location should be considered as a candidate, but in order to save space we only address such candidates a moment before their corresponding character arrives since this is the first time the algorithm can obtain proof that the candidate is not a match.}.
Notice that when a new text character arrives, the text intervals are shifted by one location.
The second component maintains all of the candidates in a given text interval. This implementation
leverages periodicity properties of strings in order to guarantee that the candidates in a given text interval form an arithmetic progression, and thus can be maintained with constant space. The third component is a fingerprint mechanism for testing if a candidate is still valid. Whenever the border of a text interval passes through a candidate, that candidate is tested.

The main challenge in applying the above framework for patterns with wildcards comes from the lack of a good notion of periodicity which can guarantee that the candidates in a text interval form an arithmetic progression.
To tackle this challenge, we design a new method for partitioning the pattern into intervals, which, combined with new fundamental combinatorial properties, leads to an efficient way for maintaining the candidates in small space. In particular, we prove that with our new partitioning there are at most $O(d\log m)$ candidates that are not part of any arithmetic progression for any text interval. Remarkably, the proof bounding the number of such candidates uses a more global perspective of the pattern, as opposed to the techniques used in non-wildcard results.

\paragraph{Overview of algorithm for Theorem~\ref{thm:tradeoff}.}
The algorithm of Theorem~\ref{thm:tradeoff} uses the algorithm of Theorem~\ref{thm:main} (with a minor adaptation) combined with a new combinatorial perspective on \emph{periodicity} that applies to strings with wildcards. The notion of periodicity in strings (without wildcards) and its usefulness are well studied~\cite{FW65, KMP77,PP09,BG14,Gawrychowski13,GS83}. However, extending the usefulness of periodicity to strings with wildcards runs into difficulties, since the notions are either too inclusive or too exclusive (see~\cite{Blanchet08,BB99,BH02,CMSWY03,SW09}). Thus, we introduce a new definition of periodicity, called the \emph{wildcard-period length} that captures, for a given pattern with wildcards, the smallest possible average distance between occurrences of the pattern in any text. See Definition~\ref{def:wildcard-period}. For a string $S$ with wildcards, we denote the wildcard-period length of $S$ by $\pi_S$.

Let $P^{*}$ be the longest prefix of $P$ such that $\pi_{P^*} \leq d^\delta$.
The algorithm of Theorem~\ref{thm:tradeoff} has two main components, depending on whether $P^* = P$ or not.
In the case where $P^* = P$, the algorithm takes advantage of the wildcard-period length of $P$ being small, which, together with techniques from number theory and new combinatorial properties of strings with wildcards, allows to spend only $\tilde{O}(1)$ time per character and uses $\tilde{O}(d^{1+\delta})$ words of space. This is summarized in Theorem~\ref{thm:smallWPAlgorithm}. Of particular interest is Lemma~\ref{lem:GammaBound} which combines number theory with combinatorial string properties in a new way. We expect these ideas to be useful in other applications.

If $P^* \neq P$, then we use the algorithm of Theorem~\ref{thm:smallWPAlgorithm} to locate occurrences of $P^*$, and by maximality of $P^{*}$, occurrences of prefixes of $P$ that are longer than $P^{*}$ must appear far apart (on average).
These occurrences are given as input to a minor adaptation of the algorithm of Theorem~\ref{thm:main} in the form of candidates.
Utilizing the large average distance between candidates, we obtain an $\tilde{O}(d^{1-\delta})$ amortized time cost per character.

\section{Preliminaries} \label{Preliminaries}

\subsection{Periods}\label{PeriodsSection}
We assume without loss of generality that the alphabet is $\Sigma=\{1,2,\ldots,n\}$.
For a string $S=s_0s_1\dots s_{\ell-1}$ over $\Sigma$ and integer $1\leq k\leq \ell$ , the substring $s_0s_1\dots s_{k-1}$ is called a \emph{prefix} of $S$ and $s_{\ell-k}\dots s_{\ell-1}$ is called a \emph{suffix} of $S$.

A prefix of $S$ of length  $i\geq 1$ is a \emph{period} of $S$ if and only if $s_j=s_{j+i}$ for every $0\leq j \leq \ell-i-1$.
The shortest period of $S$ is called \emph{the principle period} of $S$, and its length is denoted by  $\rho_S$.
If $\rho_S\leq \frac{|S|}{2}$ we say that $S$ is \textit{periodic}.

The following lemma is due to Breslauer and Galil~\cite{BG14}.
\begin{lemma}[{\cite[Lemma 3.1]{BG14}}]\label{lem:BG31}
Let $u$ and $v$ be strings such that $u$ contains at least three occurrences of $v$. Let $t_1<t_2<\dots <t_h$ be the locations of all occurrences of $v$ in $u$. Assume that $h \geq 3$ and that for $i = 1,\dots, h - 2$, we have $t_{i+2} - t_i \leq |v|$. Then, the sequence $(t_1,t_2,\dots, t_h)$  forms an arithmetic progression with difference $\rho_v$.
\end{lemma}
The following lemmas follow from Lemma 3.

\begin{lemma} \label{lem:threeOcc}
Let $v$ be a string of length $\ell$ and let $u$ be a string of length at most $2\ell$. If $u$ contains at least three occurrences of $v$ then the distance between any two occurrences of $v$ in $u$ is a multiple of $\rho_v$ and  $v$ is a periodic string.
\end{lemma}
\begin{proof}
Let $0\leq c_1<c_2<c_3 \leq |u|-1$ be three occurrences of $v$ in $u$.
Thus, $c_3\leq (|u|-1)-(|v|-1)\leq 2\ell-\ell=\ell$, and so  $c_3-c_1\leq\ell$.
Therefore, by Lemma~\ref{lem:BG31}, all the occurrences of $v$ in $u$ form an arithmetic progression with common difference $\rho_v$. In particular, the distance between any two occurrences of $v$ in $u$ is a multiple of $\rho_v$.
Hence, $\rho_v+\rho_v \leq(c_3-c_2)+(c_2-c_1)=c_3-c_1\leq \ell= |v|$ and $\rho_v\leq \frac{|v|}{2}$. Thus, by definition, $v$ is a periodic string.
\end{proof}

\begin{lemma}\label{lem:substringPeriod}
Let $u$ be a periodic string over $\Sigma$ with principle period length $\rho_u$.
If $v$ is a substring of $u$ of length 	at least $2\rho_u$ then $\rho_u=\rho_v$.
\end{lemma}
\begin{proof}
Since $v$ is a substring of $u$, we have by definition that $\rho_u$ is a period length of $v$, and thus $\rho_v\le \rho_u$ by the minimality of $\rho_v$.

It only remains to prove that $\rho_u\le \rho_v$, which we do by showing that $\rho_v$ is a period length of $u$.
We denote $u=u_0u_1\dots u_{|u|-1}$ .

Let $0\le i<|u|-\rho_v$ be an index in $u$, we have to prove that $u_i=u_{i+\rho_v}.$

Let $a$ be an index such that $v$ occurs in $u$ in position $a$, thus $u_au_{a+1}\dots u_{a+2\rho_u-1}$ is a substring of both $u$ and $v$.
Since $\rho_u$ is a period length of $u$, $u_i=u_{i+z\cdot \rho_u}$ for any $z\in \mathbb Z$ if $0\le i+z\cdot \rho_u<|u|$. In particular, for $z=\ceil{\frac{a-i}{\rho_u}}$ we have that $u_i=u_{i+\ceil{\frac{a-i}{\rho_u}}\rho_u}$.
Let $b=i+z\cdot \rho_u$. Notice that $a\le b<a+\rho_u$ and $a\le b+\rho_v<a+2\rho_u$. Therefore, $b$ and $b+\rho_v$ are both indices of characters in $v$, and thus $u_b=u_{b+\rho_v}$.
Hence, we have that $u_i=u_{i+z\cdot \rho_u}=u_{i+z\cdot \rho_u+\rho_v}=u_{i+\rho_v}$, where the last equality is based again on the fact that $\rho_u$ is a period length of $u$.
\end{proof}

\paragraph{Periods and wildcards.} For a string $u$ with no wildcards, there is an inverse relationship between the maximum number of occurrences of $u$ in a text of a given length and the principle period length of $u$. Next, we define the \emph{wildcard-period length} of a string over $\Sigma \cup \{?\}$ which captures a similar type of relationship for strings \emph{with} wildcards. The usefulness of this definition for our needs is discussed in more detail in Section~\ref{sec:tradeoff}.
Let $occ(S',S)$ be the number of occurrences of a string $S$ in a string $S'$.
\begin{definition}\label{def:wildcard-period}
For a string $S$ over $\Sigma \cup \{?\}$, the wildcard-period length of $S$ is
$$\pi_S=\min_{S'\in \Sigma^{2|S|-1} }\left \{ \ceil{\frac{|S|}{ occ(S',S)} } \right \}.$$
\end{definition}

\subsection{Fingerprints}\label{FingerprintsSection}
For the following let $u,v\in \bigcup _{i=0} ^n {\Sigma^i}$  be two strings of size at most $n$.
Porat and Porat~\cite{PP09} and Breslauer and Galil~\cite{BG14} proved the existence of a \textit{sliding fingerprint function} $\phi : \bigcup _{i=0} ^n {\Sigma^i} \rightarrow[n^c]$, for some constant $c>0$, which is a function where:

\begin{enumerate}

\item If $|u|=|v|$ and $u\neq v$ then $\phi(u)\neq\phi(v)$ with high probability (at least $1-\frac{1}{n^{c-1}}$).
\item \emph{The sliding property:} Let $w$=$uv$ be the concatenation of $u$ and $v$. If $|w|\leq n$ then given the length and the fingerprints of any two strings from $u$,$v$ and $w$, one can compute the fingerprint of the third string in constant time.

\end{enumerate}
\begin{figure}
\begin{codebox}
\Procname{$\proc{Init}()$ }
\li 	$Q_0.\Enqueue(0)$
\end{codebox}
\begin{codebox}
\Procname{$\proc{Process-Character}(t_\alpha)$ }
\li \For $h \gets 0$ \To $k$
\li		\Do	
$c \gets Q_h.\Dequeue()$
\li				\If $c$ exists and $c$ is valid
\Then
\li				\If $h=k$
\Then
\li					report $c$ as a match
\li				\Else
$Q_{h+1}.\Enqueue(c)$
\End
\End
\End
\End
\li $Q_0.\Enqueue(\alpha+1)$
\end{codebox}
\caption{Generic Algorithm. The purpose of the initialization is to consider location $0$ as a candidate before any candidate has arrived.}
\label{fig:GenericAlog}
\end{figure}

\section{A Generic Algorithm}\label{sec:generic}

We start with a generic algorithm (pseudo-code is given in Figure~\ref{fig:GenericAlog}) for solving  pattern matching problems in the streaming model. With proper implementations of the algorithm's components, the algorithm solves the PMDW problem.
The generic algorithm makes use of the notion of a \emph{candidate}.
Initially every text index $c$ is considered as a candidate for a pattern occurrence from the moment $t_{c-1}$ arrives.
An index continues to be a candidate until the algorithm encounters proof that the candidate is not a valid occurrence (or until it is reported as a match).
A candidate is \emph{alive} until such proof is given.

The generic algorithm is composed of three conceptual parts that affect the complexities of the algorithm. An example of an execution of the generic algorithm appears in Figures~\ref{fig:GenericPattern} and~\ref{fig:GenericPattern2}:
\begin{itemize}
\item{\textbf{Pattern and text intervals.}}
The first part is an ordered list $\mathcal{I}=(I_0,\dots, I_k)$ of intervals.
The disjoint union of the intervals of $\mathcal{I}$ is exactly $[0,m-1]$ and the intervals are ordered such that $I=[i,j]$ precedes $I'=[i',j']$ if and only if $j<i'$.
Each interval $I\in \mathcal{I}$ is called a \textbf{pattern interval}.
For each pattern interval $I=[i,j]\in\mathcal I$ we define a corresponding \textbf{text interval}, $\ti{I}{\alpha}=[\alpha-j+1,\alpha-i+1]$.
When character $t_\alpha$ arrives, a text location $c\in \ti{I}{\alpha}$ is a candidate  if and only if $t_c\cdots t_{c+i-1}$ matches $p_0\cdots p_{i-1}$.
The  \textbf{candidate set} $\candidates{I}{\alpha}$ is the set of text positions in $\ti{I}{\alpha}$ which are candidates right after the arrival of $t_\alpha$.

\item{\textbf{Candidate queues.}}
The second conceptual part of the generic algorithm is an implementation of a \textbf{candidate-queue} data structure.
For any interval $I\in\mathcal I$, the algorithm maintains a candidate queue $Q_I$.
At any time $\alpha$, which is the time right after $t_\alpha$  arrives, but before $t_{\alpha+1}$ arrives, $Q_I$ stores a (possibly implicit) representation of $\candidates{I}{\alpha}$.
Thus, the operations of the data structure are \emph{time-dependent}.
Candidate-queues support the following operations.
\begin{definition}
A candidate-queue for an interval $[i,j]=I\in\mathcal I$ supports the following operations at time, where $t_\alpha$ is the last text character that arrived.
\begin{enumerate}
\item{} $\Enqueue()$: add $c=\alpha -i +1$ to the candidate-queue.
\item{} $\Dequeue()$: remove and return a candidate $c = \alpha -j $, if such a candidate exists.
\end{enumerate}

\end{definition}

Since there is a bijection between pattern intervals and text intervals we say that a candidate-queue that is associated with pattern interval $I$ is also associated with the corresponding  text interval $\ti{I}{\alpha}$.

\item{\textbf{Assassinating candidates.}}
The third conceptual part addresses the following.
When a new text character arrives, all the text intervals move one position ahead, and some candidates leave some text intervals and their corresponding candidate sets.
The third conceptual part is a mechanism for testing if a candidate is valid after that candidate leaves a candidate set. This mechanism is used in order to determine if the candidate should enter the candidate-queue of the next text interval, or be reported as a match if there are no more text intervals.
\end{itemize}

The implementation of each of the three components controls the complexities of the algorithm.
Minimizing the number of intervals reduces the number of candidates leaving text intervals at a given time.
Efficient implementations of the candidate-queue operations and testing if a candidate is valid control both the space usage and the amount of time spent on each candidate that leaves an interval.
Notice that the implementations of these components may depend on each other, which is also the case in our solution.

\begin{figure}[!ht]
\begin{center}
\includegraphics[width=0.5\textwidth]{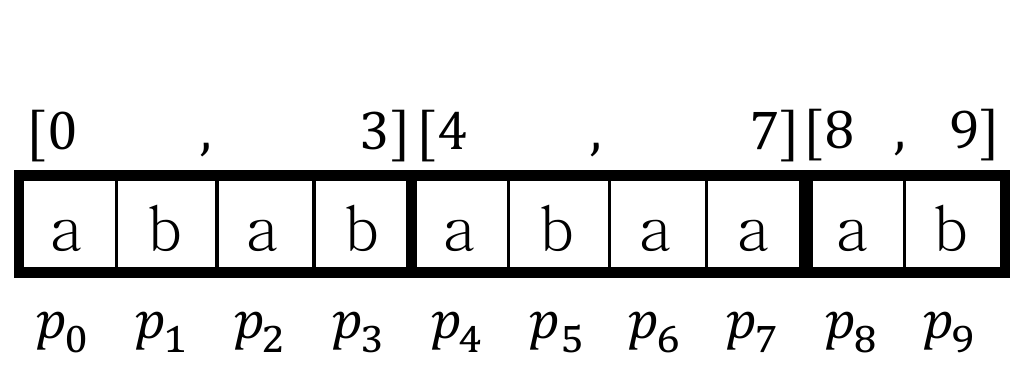}
\end{center}

\caption{Example of a pattern and its arbitrarily chosen pattern intervals. The pattern length is $10$ and the pattern intervals are $[0,3]$, $[4,7]$ and $[8,9]$.}
\label{fig:GenericPattern}
\end{figure}
\begin{figure}[!ht]

\begin{center}
\includegraphics[width=0.75\textwidth]{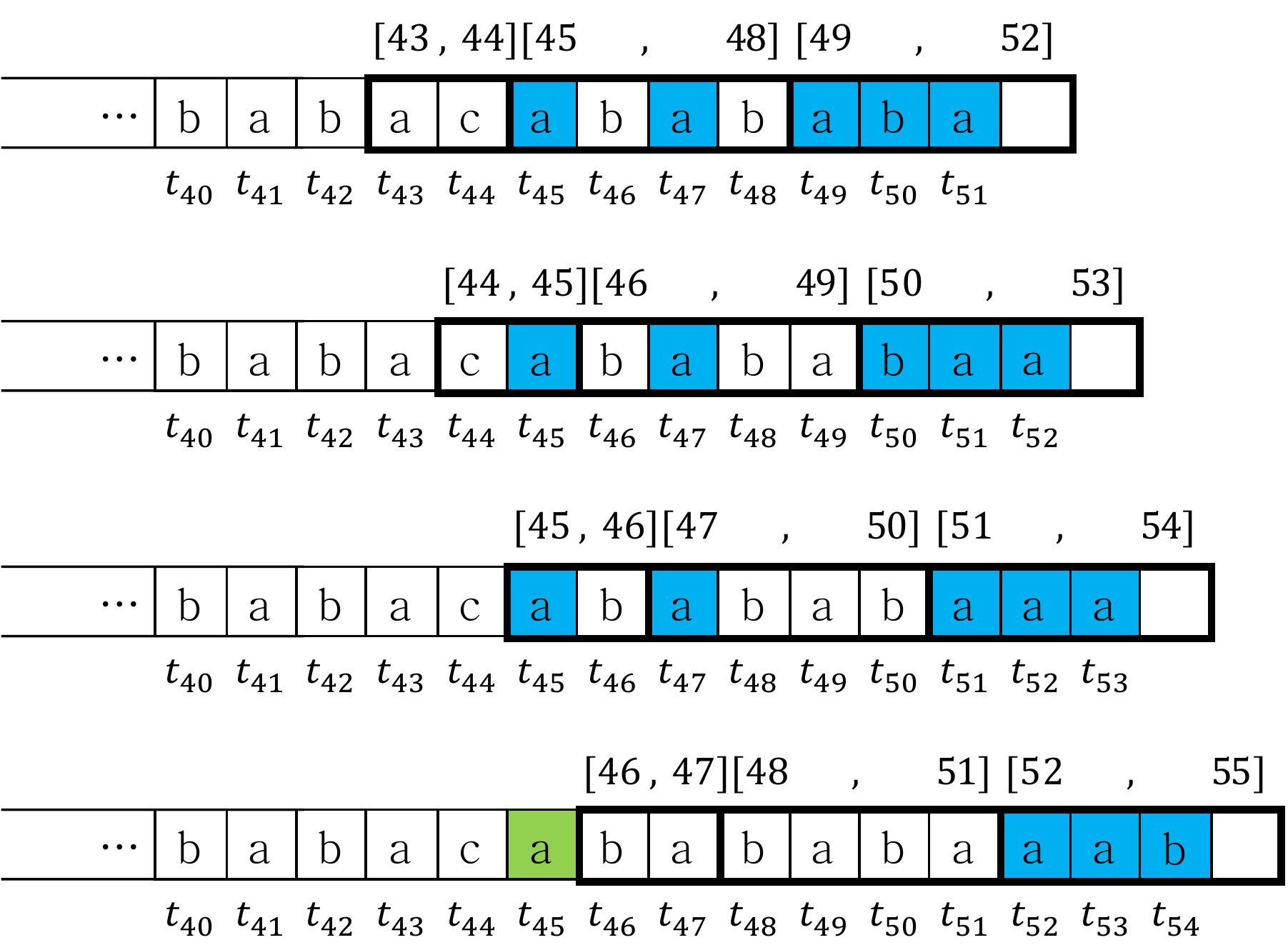}
\end{center}
\caption[LoF entry]{Example of an execution of the generic algorithm with the pattern of Figure~\ref{fig:GenericPattern}. In each row a new text character arrives. The bold borders illustrate the text intervals. Each blue cell is a position of a candidate and the green cell corresponds to a match.

When $t_{52}$ arrives, the candidate $c_1=45$ is tested, since it exits a text interval. The candidate $c_1$ remains alive because $abababaa$ is a prefix of the pattern. Notice that at this time  the candidate $c_2=47$ in not a valid occurrence of the pattern, but, the algorithm does not remove $c_2$  until $c_2$ reaches the end of the text interval.

When $t_{54}$ arrives, the candidates $c_1=45$ and $c_2=47$ are tested, as they have reached the end of their text intervals. At this time, $c_2$ is removed since the text $ababaaab$ is not a prefix of the pattern. The candidate $c_1$ remains alive and is reported as a match, since $c_1$ reached the end of the last text interval. }
\label{fig:GenericPattern2}
\end{figure}

\paragraph{A \naive implementation.}
The following \naive implementation of the generic algorithm is helpful for gaining intuition as to how the algorithm works.
Let $\mathcal{I}_{\textit{\naive}} =([0,0],[1,1],\dots,[m-1,m-1])$. The implementation of candidate queue $Q_I$  explicitly stores  the set $\candidates{I}{\alpha}$ at time $\alpha$. Notice that $\candidates{I}{\alpha}$  contains at most one candidate. The task of verifying that a candidate $c$ is valid in between text intervals is a straightforward comparison of $p_{\alpha-c}$ with $t_\alpha$. Each such comparison costs $O(1)$ time. The runtime of the algorithm is $\Theta(m)$ time per character in the worst-case, and the space usage is also $\Theta(m)$.\footnote{For example, if the pattern is $aa\dots a= a^m$ and the text is $a^n$, then each candidate $c$ is alive as long as the characters $t_c,\dots t_{c+m-1}$ arrive. Therefore, after the arrival of the first $m-1$ characters, any additional arriving character is compared with $m$ pattern characters.} We refer to this algorithm as the \emph{\naive algorithm}.

\paragraph{Using fingerprints.}
If there are no wildcards in $P$, then one can use the following fingerprint based algorithm that verifies the validity of a candidate $c$ only once all the characters $t_c,t_{c+1},\dots ,t_{c+m-1}$ have arrived.
This algorithm is closely related to the Karp and Rabin~\cite{KR87} algorithm.
The algorithm uses a partitioning of $[0,m-1]$ into only one interval containing all of $[0,m-1]$.

The algorithm maintains the \emph{text fingerprint} which is the fingerprint of the text from its beginning up to the last arriving character.
For each text index $c$, just before $t_c$ arrives the algorithm creates a candidate for the index $c$ and stores the text fingerprint $\phi(t_0t_1\dots t_{c-1})$ as satellite information of the candidate $c$. Then, $c$ (together with its satellite information) is added to the candidate-queue via the $\Enqueue()$ operation.
When the character $t_{c+m-1}$ arrives, the text fingerprint is $\phi(t_0\dots t_{c+m-1})$.
At this time, the algorithm uses the $\Dequeue()$ operation to extract $c$ together with $\phi(t_0t_1\dots t_{c-1})$ from the candidate-queue.  Then, the algorithm tests if $c$ is valid by computing $\phi(t_c\dots t_{c+m-1})$ from the current text fingerprint $\phi(t_0t_1\dots t_{c+m-1})$ and the fingerprint $\phi(t_0t_1\dots t_{c-1})$ (using the sliding property of the fingerprint function), and then testing if
$\phi(t_c\dots t_{c+m-1})$ equals $\phi(p_0\dots p_{m-1})$.
The fingerprint algorithm spends only constant time per text character, but, like the \naive algorithm, uses $\Theta(m)$ words of space to store the candidate-queue.

\subsection{Fingerprints with Wildcards}\label{sec:prelim-partition}
Using fingerprints together with wildcards seems to be a difficult task, since for any string $S$ with $x$ wildcards there are $|\Sigma|^x$  different strings over $\Sigma$ that match the string $S$. Each one of these different strings may have a different fingerprint and therefore there are $\Theta(|\Sigma|^x)$ fingerprints to store, which is not feasible.
In order to still use fingerprints for solving PMDW we use a special partitioning of $[0,m-1]$, which is described in Section~\ref{sec:second-partition}. The partitioning in Section~\ref{sec:second-partition} is based on the following preliminary partitioning.

\paragraph{The preliminary partitioning.}
We use a representation of $P$ as $P=P_0?P_1?\dots ?P_d$ where each subpattern $P_i$ contains only characters from $\Sigma$ (and may also be an empty string).
Let $W=(w_1,w_2,\dots,w_d)$ be the  indices of wildcards in $P$ such that for all $1\leq i < d$ we have $w_i<w_{i+1}$. The interval $[0,m-1]$ is partitioned into pattern intervals as follows:
$$ \mathcal{J} = ([0,w_1 -1], [w_1,w_1], [w_1+1,w_2-1] , \dots ,[w_d,w_d] ,[w_d+1,m-1]).$$
Since some of the pattern intervals in this partitioning could be empty, we discard such intervals.
The pattern intervals of the form $[w_i, w_i]$ are called \textit{wildcard interval}s and the other pattern intervals are called \textit{regular intervals}.
Notice that for a text index $c$, the substring $t_{c} \dots t_{c+m-1}$ matches $P$ if and only if for each regular interval $[i,j]$, $t_{c+i} \dots t_{c+j}=p_i\dots p_j$.

\paragraph{A preliminary algorithm.}
Given the preliminary partition $\mathcal{J}$, one could use the following algorithm for testing the validity of a candidate $c$ whenever it leaves a text interval.
During the initialization of the algorithm we precompute and store the fingerprints for all of the subpatterns corresponding to regular intervals.
Each time a candidate $c$ is added to a candidate-queue for interval $[i,j]\in \mathcal{J}$ via the $\Enqueue()$ operation, the algorithm stores the current text fingerprint $\phi(t_0\dots t_{c+i-1})$ together with the candidate $c$.
When the character $t_{c+j}$ arrives, the text fingerprint is $\phi(t_0\dots t_{c+j})$.
At this time, the algorithm uses the $\Dequeue()$ operation to extract $c$ together with $\phi(t_0t_1\dots t_{c+i-1})$ from the candidate-queue of interval $[i,j]$.
If $[i,j]$ is a regular interval, then the algorithm tests if $c$ is valid, and removes (assassinates) $c$ if it is not. This validity test is executed by applying the sliding property of the fingerprint function to compute $\phi(t_{c+i} \dots t_{c+j})$ from the current text fingerprint $\phi(t_0t_1\dots t_{c+j})$ and the fingerprint $\phi(t_0t_1\dots t_{c+i-1})$, and then testing if
$\phi(t_{c+i}\dots t_{c+j})$ is the same as $\phi(p_i\dots p_{j})$.
If $[i,j]$ is a wildcard interval then $c$ stays alive without any testing.

A \naive implementation of the candidate queues provides an algorithm that costs $O(d)$ time per character, but uses $\Theta(m)$ words of space. To overcome this space usage we employ a more complicated partitioning, which, together with a modification of the requirements from the candidate-queues, allows us to design a data structure that uses much less space. However, this space efficiency comes at the expense of a slight increase in the time per character.

\section{The Partitioning}\label{sec:second-partition}
The key idea of the new partitioning is to use the partitioning of Section~\ref{sec:prelim-partition} as a preliminary partitioning, and then perform a secondary partitioning of the regular pattern intervals, thereby creating even more regular intervals.
As mentioned, the intervals are partitioned in a special way which allows us to implement candidate-queues in a compact manner (see Section~\ref{sec:fingerprint-queues}).

The following definition is useful in the next lemma.
\begin{definition}
For an ordered set of intervals $\mathcal{I} =(I_0,I_1,\dots I_k)$
and for any integer $0\leq x \leq k$, let  $\mu_{\mathcal{I}}(x)= \max_{0\le y \le x} \left \{ |I_y| \right \}$
be the length of the longest interval in the sequence $I_0,\dots I_x$. When $\mathcal{I}$ is clear from context we simply write
$\mu(x)= \mu_{\mathcal{I}}(x) $

\end{definition}

The following lemma provides a partitioning which is used to improve the preliminary partitioning algorithm. The properties of the partitioning that are described in the statement of the lemma are essential for our new algorithm.
The most essential property is property~\ref{property:prefixSequence}, since it guarantees that  for each pattern interval $I=[i,j]$,  there exists a substring of $P$ prior to $p_i$ and with no wildcards whose length is $|I|$. If this substring is not periodic, then for any $\alpha$, $\candidates{I}{\alpha}$
does not contain more than two candidates.
If this substring is periodic,  then we show how to utilize the periodicity of the string in order to efficiently maintain all the candidates in $\candidates{I}{\alpha}$
for any $\alpha$ (see Section~\ref{sec:fingerprint-queues}).
In the proof of the lemma we introduce a specific partitioning which has all of the stated properties.

\begin{lemma} \label{lem:partitioningProperties}
Given a pattern $P$ of length $m$
with $d$ wildcards,
there exists a partitioning of the interval $[0,m-1]$ into subintervals $\mathcal{I}=(I_0,I_1 \dots, I_k)$
which has the following properties:
\begin{enumerate}
\item If $I=[i,j]$ is a pattern interval then $p_i\dots p_j$  either corresponds to exactly one wildcard from $P$ (and so $j=i$) or it is a substring that does not contain any wildcards.
\item $k=O(d+\log m)$.
\item \label{property:prefixSequence}
For each regular pattern interval $I=[i,j]$ with $|I|>1$, the length $i$ prefix of $P$ contains a consecutive sequence of $|I|$ non-wildcard characters.
\item \label{property:k-d-plus-logm} $|\{\mu_{\mathcal{I}}(0),\mu_{\mathcal{I}}(1)\dots\mu_{\mathcal{I}}(k)\}|=O(\log m)$.
\end{enumerate}
\end{lemma}
\begin{proof}
We introduce a secondary partitioning of the preliminary partitioning described in Section~\ref{sec:prelim-partition}, and prove that the secondary partitioning has all the required properties;
see
Figures~\ref{fig:PartitioningBeginning},~\ref{fig:PartitioningContinue} and~\ref{fig:ExampleOfLevels}.
Recall that we use a representation of $P$ as $P=P_0?P_1?\dots ?P_d$.
Let $J_h$ be the preliminary pattern interval corresponding to $P_h$.
The secondary partitioning is executed on the pattern intervals $\mathcal{J}=(J_0, J_1,\dots ,J_d)$, where the partitioning of $J_h$ is dependent on the partitioning of $J_0,\dots , J_{h-1}$.
Thus, for $h>0$, the secondary partitioning of $J_h$ takes place only after the secondary partitioning of $J_{h-1}$.

When partitioning pattern interval $J_h = [i,j]$, let $g_h$ be the number of pattern intervals in the secondary partitioning of $[0,i-1]$, and let  $\delta_h$ be the length of the longest pattern interval in the secondary partitioning of $[0,i-1]$. For the first pattern interval let $\delta_0=1$.
If $j\leq i+\delta_h-1$ then the only pattern interval is all of $J_h$.
If $j\leq i+2\cdot\delta_h-1$ then we create the pattern intervals $[i,i+\delta_h-1]$ and $[i+\delta_h,j]$.
Otherwise, we first create the pattern intervals $[i,i+\delta_h-1]$ and $[i+\delta_h, i+2\cdot\delta_h -1]$\footnote{The choice of having the first two intervals to be of the same length $\delta_h$ is in order to guarantee the third property in the lemma, as shown below.}, and for as long as there is enough room in the remaining preliminary pattern interval $J_h$ (between the position right after the end of the last secondary pattern interval that was just created and $j$)
we iteratively create pattern intervals where the length of each pattern interval is double the length of the previous pattern interval.
Once there is no more room left in $J_h$, let $\ell$ be the length of the last pattern interval we created. If the remaining part of the preliminary pattern interval is of length at most $\ell$, then we create one pattern interval for all the remaining preliminary pattern interval. Otherwise we create two pattern intervals, the first pattern interval of length $\ell$ and the second pattern interval using the remaining part of $J_h$.

\begin{figure}[]
\includegraphics[width=\textwidth]{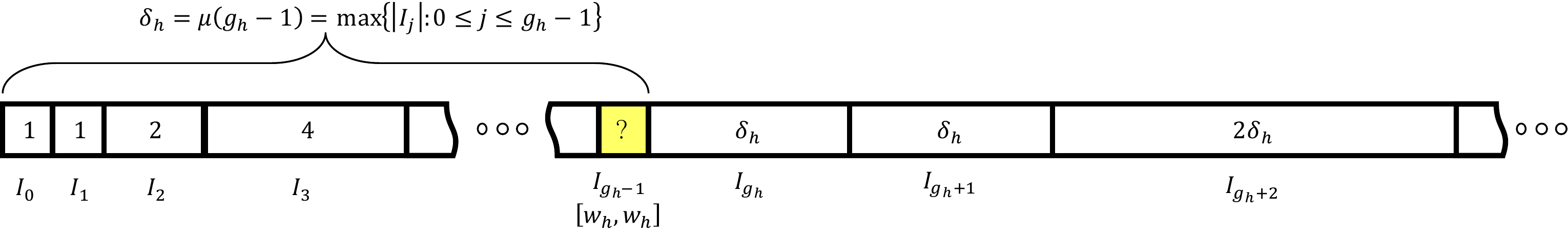}
\begin{center}
\end{center}
\caption{The general case: for each $J_h\in \mathcal{J}$ we first create two intervals of length $\delta_h$  and then
we iteratively create pattern intervals where the length of each pattern interval is double the length of the previous pattern interval.}
\label{fig:PartitioningBeginning}
\end{figure}
\begin{figure}[]
\includegraphics[width=\textwidth]{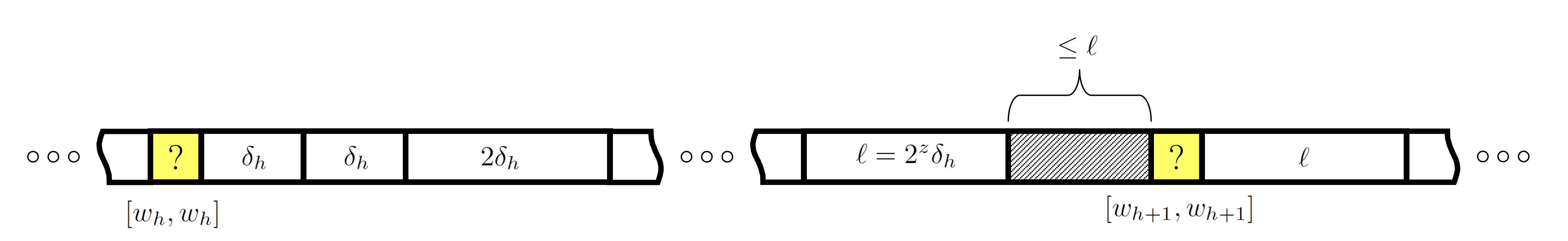}
\includegraphics[width=\textwidth]{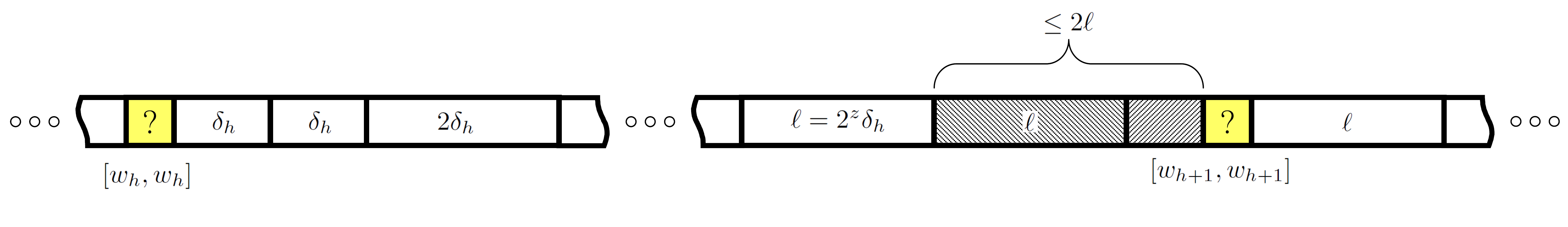}
\begin{center}
\end{center}
\caption{Once there is no more room left in $J_h$, if the remaining interval is of length at most $\ell$ (the top case), then we create one pattern interval for all the remaining interval. Otherwise (the bottom case) we create two pattern intervals, the first pattern interval of length $\ell$ and the second pattern interval using the remaining part of $J_h$.}
\label{fig:PartitioningContinue}
\end{figure}

The secondary partitioning implies all of the desired properties:

\vspace{4pt}\noindent
\textbf{Property 1.} Being that the secondary partitioning is a sub partitioning of the preliminary partitioning and the preliminary  partitioning already had this property, then the secondary partitioning has this property as well.

\vspace{4pt}\noindent
\textbf{Property 2.} For a subpattern $P_h$, the length of every pattern interval created from $J_h$ during the secondary partitioning, except for the first two pattern intervals and possibly also the last two pattern intervals, is at least twice the length of the longest pattern interval preceding it.
So the total number of such pattern intervals is $O(\log m)$.
The number of other regular pattern intervals is at most $4(d+1)$.
Additionally, there are $d$ wildcard pattern intervals.
So the total number of pattern intervals is at most $4(d+1)+d+O(\log m)= O(d+\log m)$.

\vspace{4pt}\noindent
\textbf{Property 3.} If there is a regular pattern interval $I'=[i',j']$ such that $j'<i$ and $|I'|\geq|I|$, then  the subpattern associated with $I'$ meets the requirement.

If there is no such pattern interval, it must be the case that the length of $I$ is twice the length of the pattern interval preceding $I$, and $I$ is contained in a preliminary pattern interval $J_h$ for some $h$. Let the length of the first pattern interval created in $J_h$ be denoted by $\delta_h$.
Let $I_{h,1},I_{h,2},\dots I_{h,r}$ be the first $r$ pattern intervals created in $J_h$ such that $I_{h,r}=I$. The length of any pattern interval $I_{h,r'}$ for $1<r'\leq r$ is $2^{r'-2}\delta_h$ (since $|I_{h,1}|=|I_{h,2}|=\delta_h$, and for $2<r'\leq r$ we have $|I_{h,r'}|=2|I_{h,r-1}|$), and in particular the length of $I$ is $2^{r-2}\delta_h$. Recall that $I=[i,j]$. The length of the prefix of $P_h$ up to the index $i$  is the sum of the lengths of all the pattern intervals $I_{h,r'}$ for $r'<r$. These lengths sum up to $(1+\sum_{r'=2}^{r-1}{2^{r'-2}})\delta_h=2^{r-2}\delta_h=|I|$. So the prefix of $P_h$ fulfills the requirement.

\vspace{4pt}\noindent
\textbf{Property 4.} We prove a stronger claim: for each $0\leq h \leq k$, $\mu_{\mathcal{I}}(h)$ is a power of $2$.

This is true by induction. The first pattern interval is of length $1$, and therefore $\mu_{\mathcal{I}}(0)=1=2^0$.
If $|I_h|\leq \mu_{\mathcal{I}}(h-1)$ then $\mu_{\mathcal{I}}(h)=\mu_{\mathcal{I}}(h-1)$ which is a power of $2$ by the induction hypothesis. Otherwise, if $|I_h|>\mu_{\mathcal{I}}(h-1)$ then by the secondary partitioning algorithm $|I_h|=2\mu_{\mathcal{I}}(h-1)$, and $\mu_{\mathcal{I}}(h)=2\mu_{\mathcal{I}}(h-1)$. Hence $\mu_{\mathcal{I}}(h)$ is also a power of $2$.

The largest pattern interval is at most of length $m$, and therefore there are at most $\ceil{\log m}$ different values in $\{\mu_{\mathcal{I}}(0),\mu_{\mathcal{I}}(1),\dots ,\mu_{\mathcal{I}}(k)\}$.
\end{proof}

\begin{figure}[]
\includegraphics[width=\textwidth]{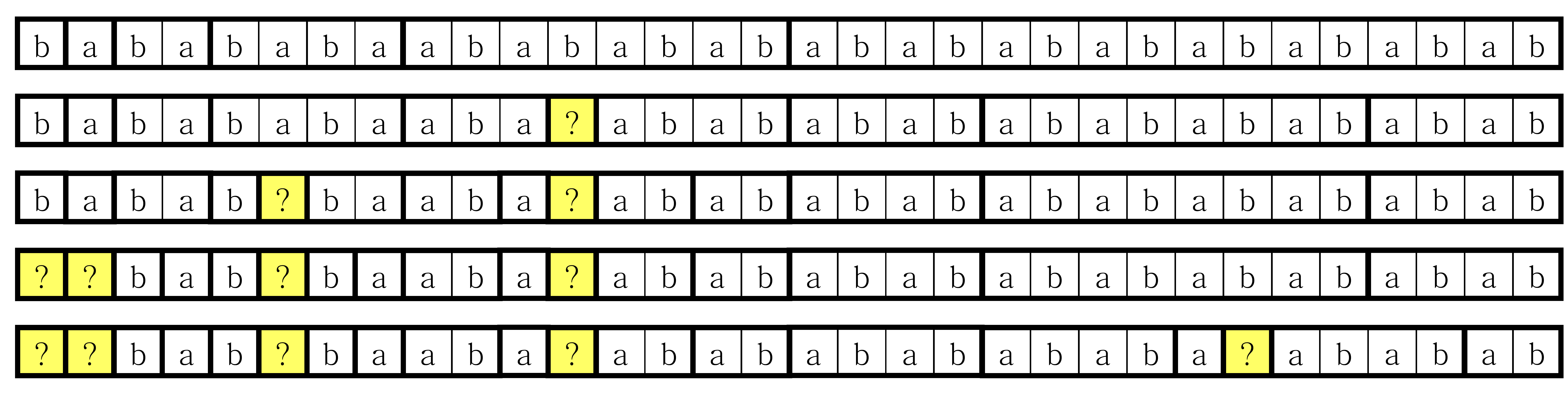}
\begin{center}
\end{center}
\caption{Example of patterns and their intervals in the secondary partitioning. Each bold rectangle corresponds to an interval in the partition.}
\label{fig:ExampleOfLevels}
\end{figure}

\section{The Candidate-fingerprint-queue}\label{sec:fingerprint-queues}
The algorithm of Theorem~\ref{thm:main} is obtained via an implementation of the candidate-queues that uses $O(d\log m)$ words of space, at the expense of having $O(d+\log m)$ intervals in the partitioning.
Such space usage implies that we do not store all candidates explicitly.
This is obtained by utilizing properties of periodicity in strings.
Since candidates are not stored explicitly, we cannot store explicit information per candidate, and in particular we cannot explicitly store fingerprints. On the other hand, we are still interested in using fingerprints in order to perform assassinations.

To tackle this, we strengthen our requirements from the candidate-queue data structure to return not just the candidate but also the fingerprint information that is needed to perform the test of whether the candidate is still valid.
For our purposes, this data structure cannot explicitly maintain all the fingerprints information.
Thus, we extend the definition of a candidate-queue to a \emph{candidate-fingerprint-queue} as follows.

\begin{definition}
A candidate-fingerprint-queue for an interval $[i,j]=I\in\mathcal I$ supports the following operations, where $t_\alpha$ is the last text character that arrived.
\begin{enumerate}

\item{} $\Enqueue(\phi(t_0 \dots t_{\alpha-i}))$: add $c=\alpha -i +1$  to the candidate-queue.
\item{} $\Dequeue()$: remove and return a candidate $c = \alpha -j $, if such a candidate exists, together with $\phi(t_0 \dots t_{c-1})$ and $\phi(t_0 \dots t_{c+i-1})$.
\end{enumerate}

\end{definition}

In order to reduce clutter of presentation, in the rest of this section we refer to the candidate-fingerprint-queue simply as the \emph{queue}.

\subsection{Implementation}
Our implementation of the queue assumes that we use a partitioning that has the properties stated in Lemma~\ref{lem:partitioningProperties}.
Let $I=[i,j]$ be a pattern interval in the partitioning and let $c$ be a candidate from $\candidates{I}{\alpha}$.
The \emph{entrance prefix} of $c$ is the substring $t_c\dots t_{c+i-1}$,  and the \emph{entrance fingerprint} is $\phi(t_c\dots t_{c+i-1})$.
By definition, since $c\in\candidates{I}{\alpha}$, the entrance  prefix of $c$ matches $p_0\dots p_{i-1}$ (which may contain wildcards).
Recall that a candidate $c$ is inserted into $Q_I$ together with $\phi(t_0\dots t_{c-1})$, which we call the \emph{candidate fingerprint} of $c$.

\paragraph{Satellite information.} The implementation associates each candidate $c$ with \emph{satellite information} (\si{}), which includes the candidate fingerprint and the entrance fingerprint of the candidate.
The \si{} of a candidate combined with the sliding property of fingerprints are crucial for the implementation of the queue.
When $c$ is added to $Q_I$, for some $I=[i,j]$, we compute the entrance fingerprint of $c$ from the candidate fingerprint and from $\phi(t_0\dots t_{c+i-1})$ which is the text fingerprint at that time.
When $c$ is removed from $Q_I$, we compute $\phi(t_0\dots t_{c+i-1})$ in constant time from the \si{} of $c$. See Figure~\ref{fig:SI-fingerprints}.

\begin{figure}[]
\begin{center}
\includegraphics[width=\textwidth]{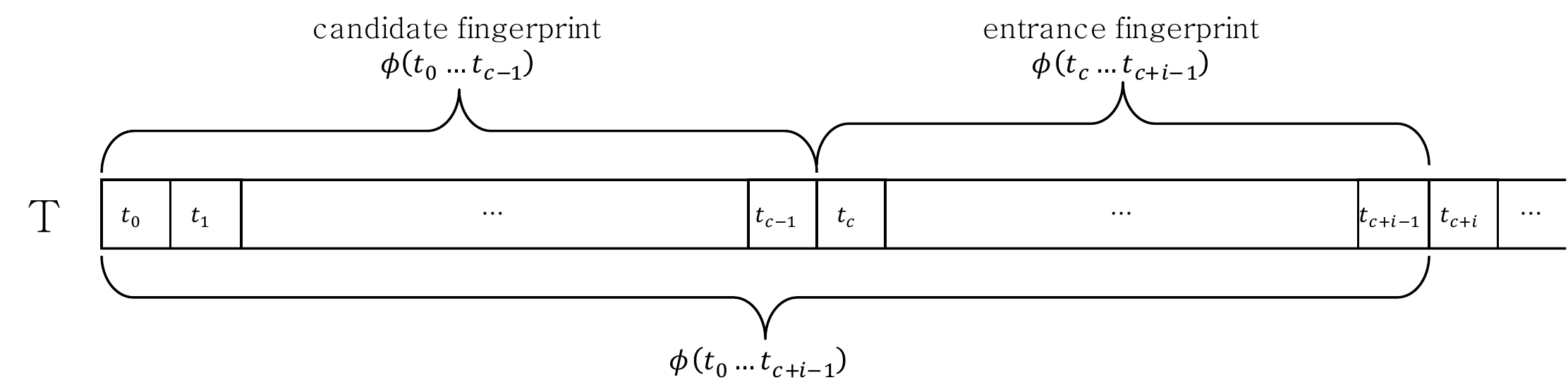}
\end{center}
\caption{The satellite information of a candidate $c$ in a text interval $\ti{I}{\alpha}$ for $I=[i,j]$ includes the candidate fingerprint $\phi(t_0\dots t_c-1)$ and the entrance fingerprint $\phi(t_c\dots t_{c+i-1})$.
The fingerprint of $\phi(t_0\dots t_{c+i-1})$ can be computed in constant time, using the sliding fingerprint property.}
\label{fig:SI-fingerprints}
\end{figure}

\paragraph{Arithmetic progressions and entrance prefixes.}

In order to implement the queue using a small amount of space, we distinguish between two types of candidates for each interval $I=[i,j]\in \mathcal I$. The first type are candidates that share a specific entrance prefix, $u_I$, which is defined solely by $p_0\dots p_{i-1}$ and is chosen such that if there are more than two candidates in $\candidates{I}{\alpha}$ with the same entrance prefix then this entrance prefix must be $u_I$ (see Lemma~\ref{lem:u_I}).
In Lemma~\ref{lem:u_Iprogression} we prove that all the candidates in $\candidates{I}{\alpha}$ that have entrance prefix $u_I$, form an arithmetic progression. This leads to Lemma~\ref{lem:uI_candidates_constant_space} where we show that all of theses candidates and their \SI{} information can be stored implicitly using $O(1)$ words of space.
The second type of candidates are the rest of the candidates, and these candidates are stored explicitly together with their \SI{} information.
We prove in Lemma~\ref{lem:boundSpace} that the total number of such candidates is $O(d\log m)$, thereby obtaining our claimed space usage.

\begin{lemma}\label{lem:u_I}
Suppose $\mathcal{I}$ is a partitioning that satisfies the properties of Lemma~\ref{lem:partitioningProperties}.
For a pattern interval $I=[i,j]\in \mathcal{I}$, there exists a string $u_I$ such that
for any text $T'$ and time $\alpha\ge 0$ the set $\candidates{I}{\alpha}$ does not contain three candidates with the same entrance prefix $u\ne u_I$.
\end{lemma}
\begin{proof}

Let $c_1<c_2<c_3$ be three different candidates in $\candidates{I}{\alpha}$  with the same entrance prefix $u$.
By property 3 of Lemma \ref{lem:partitioningProperties} there is a string $v$ of length $|I| = j-i+1$ containing only non-wildcard characters that is a substring of the length $i$ prefix of $P$.

Let $r$ be an arbitrary location of $v$ in $p_0\dots p_{i-1}$ (since $v$ could appear several times in the prefix).
The three candidates imply that after a shift of $r$ characters from the candidates' locations, there are three occurrences of $v$ in the text. These occurrences are within a substring of the text of length at most $2|v|$, since all three candidates are in $\candidates{I}{\alpha}$ and so the distance between the first and last occurrence is at most $|I|-1=|v|-1$ (the 2 factor accommodates the full occurrence of the third $v$).
Thus, by Lemma~\ref{lem:threeOcc}, $v$ must be periodic, and $|v|\ge 2\rho_v$.

Since $c_1$, $c_2$, and $c_3$ are all occurrences of $u$ then $c_3-c_2$ and $c_2-c_1$ are period lengths of $u$. Thus, $\rho_u \leq \min \{c_2 -c_1,c_3- c_2\} \leq \frac{c_3-c_1}{2}\le \frac{(\alpha-i+1)-(\alpha-j+1)}{2} \leq \frac{j-i}{2} < \frac{|I|}2= \frac{|v|}2.$
Therefore, by  Lemma~\ref{lem:substringPeriod}, $\rho_u=\rho_v$.
Similarly, let $\alpha' > \alpha$ and suppose there are three candidates $c_4,c_5,c_6$ in $\candidates{I}{\alpha'}$. Notice that it is possible that $c_1,c_2$ and $c_3$ are not in $\candidates{I}{\alpha'}$ since it  is possible that enough time has passed for them to leave.  Suppose $c_4, c_5$ and $c_6$ share the same entrance prefix $u'$. Then $\rho_{u'}= \rho_v=\rho_{u}$.

Assume by contradiction that  $u'\neq u$. Notice that the only possible locations of mismatches between $u$ and $u'$ are the positions of wildcards in the $i$ length prefix of $P$, since both $u$ and $u'$ match this prefix. In particular, $v$ occurs in the $r$'th location of both $u$ and $u'$.
Let $k$ be an index of a mismatch between $u$ and $u'$. In particular, let the $k$'th character of $u$ be $x$, and the $k$'th character of $u'$ be $x'\ne x$.
Let $\gamma$ be an integer (possibly negative) such that the $k+\gamma\cdot \rho_v$ location in $u$ is within the occurrence of $v$ in $u$ (and so also within the occurrence of $v$ in $u'$). Notice that such a $\gamma$ must exist since $|v|\ge 2\rho_v$. 		
Since $\rho_{u'}= \rho_v=\rho_{u}$, the character at location  $k+\gamma\cdot \rho_v$ in $u$ must be $x$, while the character at location  $k+\gamma\cdot \rho_v$ in $u'$ must be $x'$. But $u$ and $u'$ match at all of the locations corresponding to $v$. Thus we have obtained a contradiction, and so $u=u'$ is unique, as required.
\end{proof}

\begin{lemma}\label{lem:u_Iprogression}
Suppose $\mathcal{I}$ is a partitioning that satisfies the properties by Lemma~\ref{lem:partitioningProperties}.
For a pattern interval $I=[i,j]\in\mathcal{I}$ and time $\alpha\ge 0$ if there are $h\geq 3$ candidates $c_1<c_2< \dots <c_h$ in $\candidates I \alpha$ that have $u_I$ as their entrance prefix, then the sequence $c_1,c_2,\ldots,c_h$ forms an arithmetic progression whose difference is $\rho_{u_I}$.
\end{lemma}

\begin{proof}
The distance between any two candidates in $\candidates I \alpha$ is at most $|I|$, and $|I|\leq i$ by Property~\ref{property:prefixSequence} of Lemma~\ref{lem:partitioningProperties}.
Hence, by Lemma~\ref{lem:BG31}, all of the occurrences of $u_I$ in $T$ that begin in $\ti I \alpha$
form an arithmetic progression with difference $\rho_{u_I}$.
Each of these occurrences matches the $i$ length prefix of $P$, and therefore is a candidate in $\candidates{I}{\alpha}$.
Hence, all the candidates of $\candidates{I}{\alpha}$ with $u_I$ as their entrance prefix form an arithmetic progression with difference of $\rho_{u_I}$.
\end{proof}

\paragraph{Implementation details.}
For any pattern interval $I=[i,j]$ and time $\alpha$ we split the set of candidates $\candidates{I}{\alpha}$ into two disjoint sets.
The set $\candidatesUi{I}{\alpha}=\{c\in\candidates{I}{\alpha}\,|\,t_c\dots t_{c+i-1}=u_I\}$ contains all the candidates whose entrance prefix is $u_I$, and the set $\candidatesNonUi{I}{\alpha}=\candidates{I}{\alpha}\setminus\candidatesUi{I}{\alpha}$ contains all the other candidates of $\candidates{I}{\alpha}$.
We use a linked list $\mathcal{L}_{Q_I}$ to store all of the candidates of $\candidatesNonUi{I}{\alpha}$ together with their \si{}.
Adding and removing a candidate that belongs in $\mathcal{L}_{Q_I}$ together with its \si{} is straightforward.
The candidates of $\candidatesUi{I}{\alpha}$  are maintained using a separate data structure that leverages Lemmas~\ref{lem:u_I} and~\ref{lem:u_Iprogression}.
Thus, during a $\Dequeue()$ operation, the queue verifies if the candidate to be returned is in $\mathcal{L}_{Q_I}$ or in the separate data structure for the $\candidatesUi{I}{\alpha}$ candidates. Finally, for each pattern interval $I$ the data structure stores the fingerprint of the the principle period of $u_I$.

\begin{lemma}\label{lem:uI_candidates_constant_space}
There exists an implementation of candidate-fingerprint-queues such that the queue $Q_I$  at time $\alpha>0$ maintains all the candidates of  $\candidatesUi{I}{\alpha}$ and their \si{} using $O(1)$ words of space.
\end{lemma}
\begin{proof}
If $|\candidatesUi{I}{\alpha}|\le 2$ then $Q_I$ stores the candidates of $\candidatesUi{I}{\alpha}$ explicitly in $O(1)$ words of space.
Otherwise, by Lemma~\ref{lem:u_Iprogression}, all the candidates of $\candidatesUi{I}{\alpha}$ form an arithmetic progression.
An arithmetic progression of arbitrary length can be represented using $O(1)$ words of space.
However, $Q_I$ also needs access to the \si{} for the candidates in this progression.
To do this, $Q_I$ explicitly stores the first candidate ($\min \candidatesUi{I}{\alpha}$) together with its \si{}, the common difference of the  progression ($\rho_{u_I}$), the length of the current progression, and the fingerprint of the principle period of $u_I$.
When a new candidate $c$ with entrance fingerprint $\phi(u_I)$ enters $Q_I$, $c$ becomes the largest element in $\candidatesUi{I}{\alpha}$, and so we first increment the length of the arithmetic progression, and if $c$ is currently the only candidate in the arithmetic progression, then $Q_I$ stores $c$ and its \si{} (since then $c$ is the first candidate in the progression).
When a $\Dequeue()$ operation needs to remove the first candidate $c$ in the progression, then $Q_I$ removes $c$, which is stored explicitly together with its \si{}, decrements the length of the progression, and if there are remaining candidates in the progression then $Q_I$ computes the information for the new first remaining candidate in order to store its information explicitly.
To do this, $Q_I$ first computes the location of the new first candidate from $\rho_{u_I}$ and the location of $c$. The \si{} of the new first candidate is computed in constant time (via the sliding property) from the fingerprint of the principle period of $u_I$ and the candidate fingerprint of $c$.
\end{proof}

\paragraph{Space usage.}
The space usage of all of the queues has three components.
The first component is the lists $\mathcal{L}_{Q_I}$, which maintains the candidates of $\candidatesNonUi{I}{\alpha}$ for all the intervals $I$.
The second component is the data structures for storing the candidates with entrance prefix $u_I$ (the candidates of $\candidatesUi{I}{\alpha}$) in each $I\in \mathcal I$.
Since, by Lemma~\ref{lem:uI_candidates_constant_space}, for each $I\in\mathcal I$ all the candidates with entrance prefix $u_I$ are maintained using $O(1)$ words, all such candidates use $O(|\mathcal I|)=O(d+\log m)$ words of space.
The third component is storing for each pattern interval $I$ the fingerprint of the the principle period of $u_I$, which takes a total of $O(d+\log m)$ words of space.
In the following lemma we prove that the total space usage of all of the lists $\mathcal{L}_{Q_I}$ is $O(d\log m)$.

\begin{lemma} \label{lem:boundSpace}
$\sum_{I\in \mathcal{I} } \left|\candidatesNonUi{I}{\alpha}\right| = O(d\log m)$.
\end{lemma}

\begin{proof}

By Lemma~\ref{lem:partitioningProperties}, we know that $|\{\mu(0),\dots ,\mu(k)\}|=O(\log m)$.
For each  $\ell \in \{\mu(0),\dots ,\mu(k)\}$ let $\mathcal{I}_\ell\subseteq \mathcal{I}$ be the sequence of all pattern intervals $I\in \mathcal{I}$ such that $\mu(I)=\ell$.
We show that $\sum_{I\in \mathcal{I}_\ell } |\candidatesNonUi{I}{\alpha}|  = O(|\mathcal{I}_\ell|+d)$.
Combining with property~\ref{property:k-d-plus-logm} of Lemma~\ref{lem:partitioningProperties} which states that $\sum_{\ell \in \{\mu(0),\dots ,\mu(k)\}}{|\mathcal{I}_\ell|} =|\mathcal{I}|= O(d+\log m)$ we have that:

\begin{align*}
\sum_{I\in \mathcal{I} } \left|\candidatesNonUi{I}{\alpha}\right|   &=\sum_{\ell \in \{\mu(0),\dots ,\mu(k)\}} \sum_{I\in \mathcal{I}_\ell } \left|\candidatesNonUi{I}{\alpha}\right|  \\[0.5ex]
&= \sum_{\ell \in \{\mu(0),\dots ,\mu(k)\}} O(|\mathcal{I}_\ell|+d)\\[0.5ex]
&	=\sum_{\ell \in \{\mu(0),\dots ,\mu(k)\}} O(|\mathcal{I}_\ell|) \,+\,\sum_{\ell \in \{\mu(0),\dots ,\mu(k)\}} O(d)\\[1ex]
&=O(d+\log m)+O(d\log m)\\[2ex]
&= O(d\log m)
\end{align*}

We focus on intervals for which $\left|\candidatesNonUi{I}{\alpha}\right| \geq 3$, since if $\left|\candidatesNonUi{I}{\alpha}\right|\le 2$ the bound is straightforward.

Let $[i^*,j^*]$ be the leftmost interval in $\mathcal{I}_\ell$. By definition of $\mathcal{I}_\ell$, we have $j^*-i^*+1=\ell$, and so by Property~\ref{property:prefixSequence} of Lemma~\ref{lem:partitioningProperties}, there exists a string $v$ of length $\ell$ containing only non-wildcard characters that is a substring of the length $i^*$ prefix of $P$.
Let $r$ be an arbitrary location of $v$ in $p_0\dots p_{i^*-1}$ (since $v$ could appear several times in the prefix).
For any $[i',j']=I'\in \mathcal{I}_\ell$ the entrance prefix (which does not contain wildcards) of each candidate in $\candidates{I'}{\alpha}$ matches the $i'$ prefix of $P$ (which can contain wildcards), and in particular, the location which is $r$ locations to the left of any candidate in $\candidates{I'}{\alpha}$ is a location of an occurrence of $v$ in the text\footnote{Notice that this occurrence is well defined since $i'\ge i^*\ge r+|v|$.}.

Since we focus on intervals $I\in\mathcal I_\ell$ for which $|\candidatesNonUi{I}{\alpha} | \geq 3$, then there exist three occurrences of $v$ in the text in  positions corresponding to a shift of $r$ characters from locations of $I$'s candidates.
These occurrences are within a substring of the text of length at most $2|v|$, since all three candidates are in $\candidates{I}{\alpha}$ and so the distance between the first and the last candidates is at most $|I|-1\leq \ell-1=|v|-1$. 	
Thus, by Lemma~\ref{lem:threeOcc}, $v$ must be periodic, and the distance between any two candidates in $\candidates{I}{\alpha}$ must be a multiple of $\rho_v$.

\begin{figure}[]
\begin{center}
\includegraphics[width=0.8\textwidth]{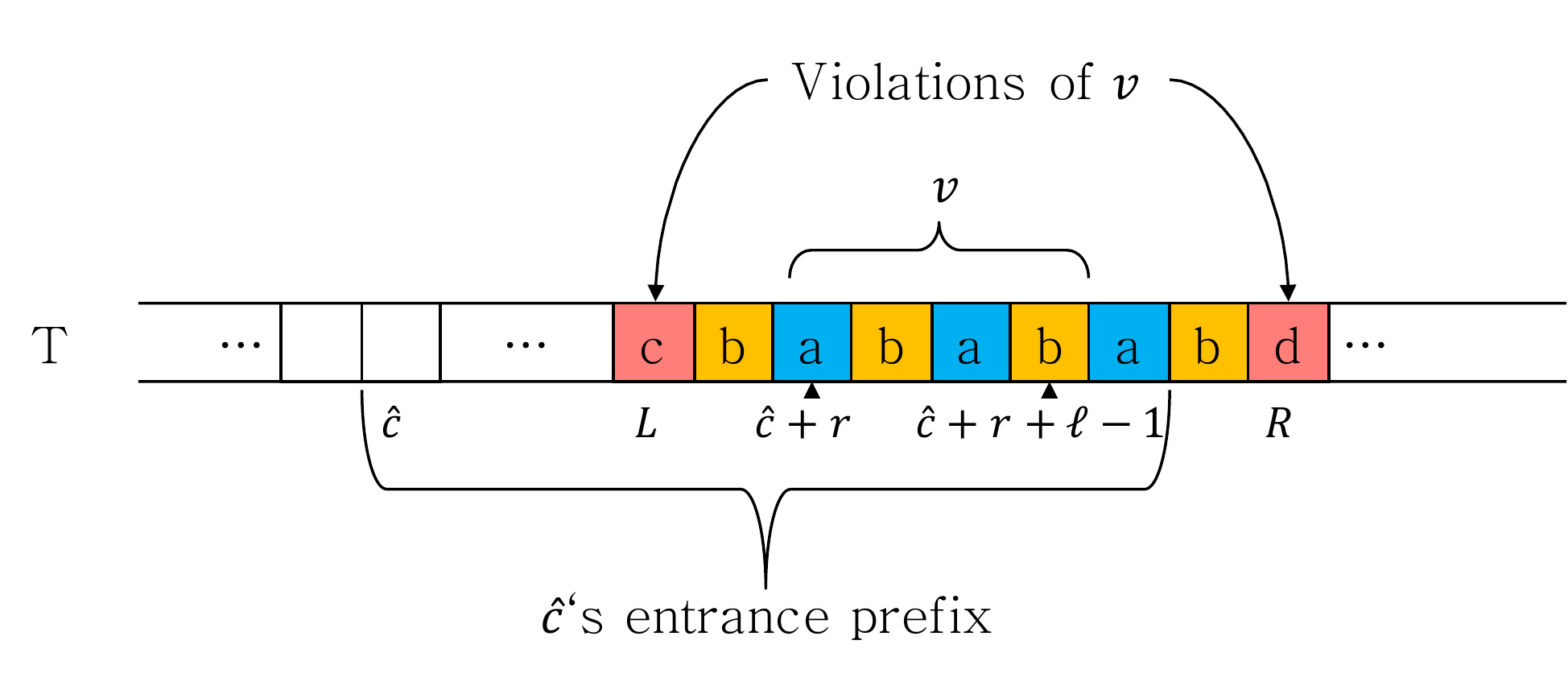}
\end{center}
\caption{Positions $L$ and $R$ are the violations of the periodic substring that contains $v$.
Notice that it is possible that $L\le \hat c$, and similarly it is possible that $R$ is in the entrance interval of $\hat c$.  }
\label{fig:left-right-violations}
\end{figure}

Let $\hat c=\max \left[\bigcup_{I\in\mathcal I_\ell}{\candidates{I}{\alpha}} \right]$ be the rightmost (largest index) candidate in the intervals corresponding to pattern intervals in $\mathcal I_\ell$.
Since $\hat c$ is a candidate in some $\candidates{I'}{\alpha}$ for $I'\in\mathcal I_\ell$, then there is an occurrence of $v$ at location $\hat c+r$. Thus,  $t_{\hat{c}+r} \dots  t_{\hat{c}+r+\ell-1}=v$.
We extend this occurrence of $v$ to the left and to the right in $T$ for as long as the length of the period does not increase.
Let the resulting substring be $t_{\lv+1}\dots t_{\rv-1}$. See Figure~\ref{fig:left-right-violations}.
If $\lv\ge 0$ then the index $\lv$ is called the \emph{left violation} of $v$.
Similarly, if $\rv\le \alpha$ then the index $\rv$ is called the \emph{right violation} of $v$.
Notice that the period of $v$ extends all the way to the beginning of the text if and only if $\lv=-1$, in which case there is no left violation.
Similarly, the period of $v$ extends all the way to the current end of the text if and only if $\rv=\alpha+1$, in which case there is no right violation.
Finally, notice that $\lv< \hat c +r\leq \hat{c}+r+\ell-1 < \rv$, since $v$ is a substring of $t_{\lv+1}\dots t_{\rv-1}$.

For a candidate $c\in\ti{[i,j]}{\alpha}$  we define the \emph{entrance interval} of $c$ to be $[c,c+i-1]$.
In addition we denote $e_c=c+i-1$, so the entrance interval of $c$ is $[c,e_c]$.

\begin{claim}\label{clm:char-overlap}
For any candidate $c\in\candidates{I}{\alpha}$  where $I\in\mathcal I_\ell$ we have $\hat c+r\in [c,e_c]$.
\end{claim}
\begin{proof}
Let $c$ be a candidate in $\candidates{I}{\alpha}$ for $I=[i,j]\in\mathcal I_\ell$.
Recall that $\candidates{I}{\alpha}\subseteq \ti{I}{\alpha}=[\alpha-j+1,\alpha-i+1]$.
Since $c$ is in this  interval  we have that $\alpha-j+1\le c\le \alpha-i+1$.
In particular, $e_c=c+i-1\ge \alpha-j+1+i-1=\alpha-(j-i+1)+1=\alpha-|I|+1$.
By definition, since $I\in\mathcal I_\ell$, we have that $|I|\le \ell$ and so $e_c\ge \alpha-\ell+1$.
Since $t_{\hat{c}+r} \dots  t_{\hat{c}+r+\ell-1}=v$, it must be that $\hat c+r+\ell-1\le \alpha$.
Thus, $\hat c+r\le \alpha-\ell+1\le e_c$.
By the maximality of $\hat c$, it is obvious that $c\le\hat c\le \hat c+r$.
Hence, we have that $c\le \hat c+r\le e_c$.
\end{proof}

\begin{claim}\label{clm:x-y}
Suppose $I=[i,j]\in\mathcal I_\ell$ and $|\candidatesNonUi{I}{\alpha}|\ge 3$. Then
for any candidate $c\in\candidatesNonUi{I}{\alpha}$ in $\normalfont\ti{I}{\alpha}$ either $\lv\in [c,e_c]$ or $\rv\in [c,e_c]$.
\end{claim}

\begin{proof}
For $c\in\candidatesNonUi{I}{\alpha}$ let $u=u_0\dots u_{i-1}$ be the entrance prefix of $c$.
Recall that $\lv< \hat c +r< \rv$.
By Claim~\ref{clm:char-overlap} it must be that $c\leq \hat c +r \leq e_c$ and so we cannot have both $\lv,\rv < c$ or both $\lv,\rv > e_c$.

Assume by contradiction that $\lv<c\leq e_c<\rv$.
We claim that there exists a text input $T'$ such that if we execute the algorithm with $T'$ as the text, then there exists some time $\beta$ where $\candidates{I}{\beta}$ contains three candidates with $u$ as their entrance prefix. Then, by Lemma~\ref{lem:u_I} we deduce that $u=u_{I}$, in contradiction to the definition of $\candidatesNonUi{I}{\alpha}$.

Recall that the principle period length of $t_{\lv+1}\dots t_{\rv-1}$ is $\rho_v$.
Since $u = t_c\dots t_{e_c}$ is a substring of $t_{\lv+1}\dots t_{\rv-1}$, it must be that $\rho_u\leq \rho_v$.
Recall that $\candidatesNonUi{I}{\alpha}$ contains at least three candidates.
Let $c_1$, $c_2$, and $c_3$ be three distinct candidates in $\candidatesNonUi{I}{\alpha}$.
Since $c_1$, $c_2$, and $c_3$ are all occurrences of $u$ then $c_3-c_2$ and $c_2-c_1$ are period lengths of $u$. Thus, $\rho_u \leq \min \{c_2 -c_1,c_3- c_2\} \leq \frac{c_3-c_1}{2} \leq \frac{j-i}{2} < \frac{|I|}2= \frac{|v|}2.$
Therefore, by  Lemma~\ref{lem:substringPeriod}, $\rho_u=\rho_v$.
Thus, $\rho_u \leq \frac{j-i}{2}$	implying that $i+2\rho_u - j\leq 0$.

Consider a long enough (at least $i+2\rho_u-1$) text $T'$ which is composed of repeated concatenation of $u_0\dots u_{\rho_u-1}$.
Notice that the substrings of $T'$ of length $i$ starting at locations $0$, $\rho_u$ and $2\rho_u$ are all exactly the string $u$, which matches $p_0\dots p_{i-1}$.
Consider an execution of the algorithm with $T'$ as the input text, and at time $\beta=i+2\rho_u-1$ consider the set $\candidates{I}{\beta}$.
We have that $\ti{I}{\beta}=[i+2\rho_u-1-j+1,i+2\rho_u-1-i+1]=[i+2\rho_u-j,2\rho_u]$.
Being that $i+2\rho_u - j\leq 0$ then the interval $[0,2\rho_u]$ is a subinterval of $\ti{I}{\beta}$, then $0$, $\rho_u$ and $2\rho_u$ are all within this interval. Thus, these locations are candidates in $\candidates{I}{\beta}$ with $u$ as their entrance prefix.
Thus, by Lemma~\ref{lem:u_I}, it must be that $u=u_I$, which contradicts $c\in\candidatesNonUi{I}{\alpha}$.
\end{proof}

Let $\candidatesNonUix{I}{\alpha}$  be the set of candidates in $\candidatesNonUi{I}{\alpha}$ whose entrance interval contains $\lv$, and let $\candidatesNonUiy$  be the set of candidates in $\candidatesNonUi{I}{\alpha}$ whose entrance interval contains $\rv$.
$\candidatesNonUix{I}{\alpha}$ and $\candidatesNonUiy$ are not necessarily disjoint.
Notice that by Claim~\ref{clm:x-y}, $\candidatesNonUix{I}{\alpha}\cup\candidatesNonUiy$ contains all the candidates of $\candidatesNonUi{I}{\alpha}$.

\begin{claim} \label{clm:bound-xy-sum}
$\sum_{I\in \mathcal{I}_\ell} \left|\candidatesNonUix{I}{\alpha}\right|=O(|\mathcal{I}_\ell| + d)$  and $\sum_{I\in \mathcal{I}_\ell} \left|\candidatesNonUiy\right|=O(|\mathcal{I}_\ell| + d)$.
\end{claim}
\begin{proof}
Let $I\in \mathcal{I}_\ell$  and let $\approx$ denote the match relation between symbols in $\Sigma \cup \{?\}$.

Notice that the contribution to $\sum_{I\in \mathcal{I}_\ell} \left|\candidatesNonUix{I}{\alpha}\right|$
from all sets $\candidatesNonUix{I}{\alpha}$ that have less than two candidates is at most $O(|\mathcal{I}_\ell|)$.
Thus, we will prove that for any set $\candidatesNonUix{I}{\alpha}$ with at least two candidates, it must be that for any candidate $c\in \candidatesNonUix{I}{\alpha}$, except for possibly one candidate, we have that $p_{\lv-c}$ is a wildcard.

Suppose $\candidatesNonUix{I}{\alpha}$ contains at least two candidates and let $c_{\textit{left}} = \max \candidatesNonUix{I}{\alpha} $ be the most recent candidate in $\candidatesNonUix{I}{\alpha}$.
Let $c<c_{\textit{left}}$ be a candidate in $\candidatesNonUix{I}{\alpha}$.
Since $c\in \candidatesNonUix{I}{\alpha}$  we have that $p_{\lv-c}\approx t_{c+\lv-c}=t_\lv$ (recall that both $\lv$ and $c$ are indices in the text).
Similarly, since $c_{\textit{left}}\in\candidatesNonUix{I}{\alpha}$  we have that $p_{\lv-c}\approx t_{c_{\textit{left}}+\lv-c}=t_{\lv+(c_{\textit{left}}-c)}$.
Recall that the distance between any two candidates in $\candidates{I}{\alpha}$ is a multiple of $\rho_v$, since $\candidates{I}{\alpha}$ contains at least $3$ candidates.
In particular the distance $(c_{\textit{left}}-c)$ is a multiple of $\rho_v$ and  $(c_{\textit{left}}-c)\leq|I|\leq|v|$.  Thus, $t_\lv \neq t_{\lv+(c_{\textit{left}}-c)}$ since $\lv$ violates the period of length $\rho_v$.
Recall that $t_\lv \approx p_{\lv-c}\approx t_{\lv+(c_{\textit{left}}-c)}$, and so $p_{\lv-c}$ must be a wildcard.
Therefore, each $c\in \candidatesNonUix{I}{\alpha}$, except for possibly $c_{\textit{left}}$, is in a position $c$ such that $p_{\lv-c}$ is a wildcard.
Since $\lv$ is the same for all of the candidates in all of the $\candidatesNonUix{I'}{\alpha}$ for all $I'\in \mathcal{I}_\ell$, then  the contribution to $\sum_{I\in \mathcal{I}_\ell} |\candidatesNonUix{I}{\alpha}|$ of the candidates that are not the most recent in their set $\candidatesNonUix{I}{\alpha}$ is at most $d$.
The contribution of the most recent candidates is at most $O(|\mathcal{I}_\ell|)$. Thus,
$\sum_{I'\in \mathcal{I}_\ell} \left|\candidatesNonUix{I'}{\alpha}\right|=O(|\mathcal{I}_\ell| + d)$.
The proof that $\sum_{I'\in \mathcal{I}_\ell} |\candidatesNonUiy|=O(|\mathcal{I}_\ell| + d)$ is symmetric.
\end{proof}

Finally, $\sum_{I\in \mathcal{I}_\ell } \left|\candidatesNonUi{I}{\alpha}\right| \leq \sum_{I\in \mathcal{I}_\ell} \left|\candidatesNonUix{I}{\alpha}\right| + \sum_{I\in \mathcal{I}_\ell} \left|\candidatesNonUiy\right|=O(|\mathcal{I}_\ell| + d)$.
Thus, we have completed the proof of Lemma~\ref{lem:boundSpace}.

\end{proof}

\section{The Algorithm of Theorem~\ref{thm:tradeoff}}\label{sec:tradeoff}
The algorithm of Theorem~\ref{thm:main} for PMDW uses $\tilde{O}(d)$ time per character and $\tilde{O}(d)$ words of space.
In this section we introduce the algorithm of Theorem~\ref{thm:tradeoff} which extends this result for a parameter $0\leq \delta \leq 1$ to an algorithm that  uses $\tilde{O}(d^{1-\delta})$ time per character and $\tilde{O}(d^{1+\delta})$ words of space.

An overview of a slightly modified version (for the sake of intuition) of the tradeoff algorithm is described as follows. Let $P^{*}$ be the longest prefix of $P$ such that $\pi_{P^*} \leq d^\delta$. The tradeoff algorithm first finds all the occurrences of $P^{*}$ in $T$ using a specialized algorithm for patterns with bounded wildcard-period length.
If $P^{*}=P$ then this completes the tradeoff algorithm.
Otherwise, let $I=[i,j]$ be the interval in the secondary partitioning of Theorem~\ref{thm:main} such that $i\leq |P^*|-1\leq j$.
We first divide $I$ into two new intervals $[i,|P^*|-1]$ and $[|P^*|,j]$. If $[|P^*|,j] =\emptyset$ then we discard $[|P^*|,j]$.
It is straightforward to see that the properties of partitions that we define in Lemma~\ref{lem:partitioningProperties} are still satisfied.
Let $I^*=[i^*=|P^*|,j^*]$ be the interval immediately following $[i,|P^*|-1]$.
Each occurrence of $P^{*}$ in the text is inserted into the algorithm of Theorem~\ref{thm:main} as a candidate directly into $Q_{I^*}$.
Thus, the entrance prefixes of candidates in the queues match prefixes of $P$ that are longer than $P^{*}$ and, by maximality of $P^{*}$, these prefixes of $P$ have large wildcard-period length.
This implies that the average distance between two consecutive candidates that are occurrences of $P^*$ is at least $d^\delta$, and so, combined with a carefully designed scheduling approach for verifying candidates, we are able to obtain an $\tilde{O}(d^{1-\delta})$ amortized time cost per character.

\paragraph{Overview.}
In Section~\ref{sec:smallPiAlgo} we describe the specialized algorithm for dealing with patterns whose wildcard-period length is at most $\tau$, for some parameter $\tau>1$.
In Section~\ref{sec:proof-tradeoff} we complete the proof of Theorem~\ref{thm:tradeoff} by describing the missing details for the tradeoff algorithm. In particular, the proof of Theorem~\ref{thm:tradeoff} uses the  algorithm of Section~\ref{sec:smallPiAlgo} with $\tau=d^\delta$.

\subsection{Patterns with Small Wildcard-period Length} \label{sec:smallPiAlgo}

\begin{figure}
\centering
\subfigure[The matrix $M^q$]
{
\includegraphics[width=0.4\textwidth]{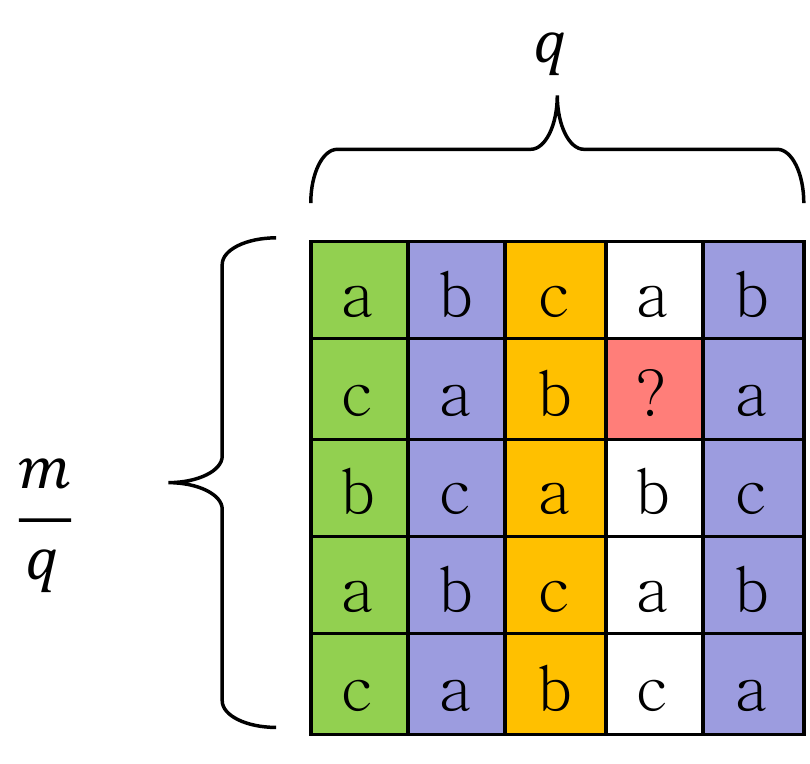}
\label{fig:Mq_matrix}
}\hfill
\subfigure[The offset patterns and $\Gamma_q$]
{
\includegraphics[width=0.4\textwidth]{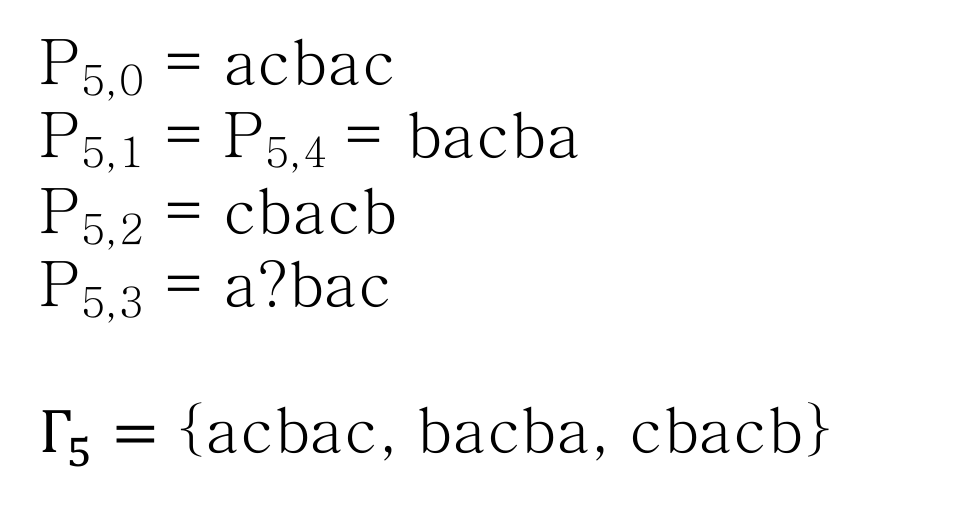}
\label{fig:Gammaq}
}
\subfigure[The column pattern $P_q$]
{
\includegraphics[width=0.2\textwidth]{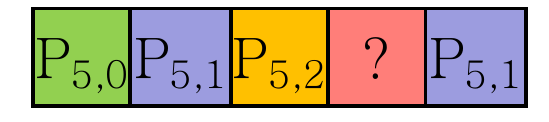}
\label{fig:Pq}
}
\caption{Example of the matrix representation for pattern $P=abcab?abcabcabcabcabc$ and $q=5$. Each color represents a unique offset pattern.
The offset patterns $P_{5,1}$ and $P_{5,4}$ are equal and therefore they have the same id (column color).
Since $P_{5,3}$ contains a wildcard, it is not associated with any id.}
\label{fig:Mq}
\end{figure}

Let $P$ be a pattern of length $m$ with $d$ wildcards  such that  $\pi_P<\tau$.
Let $q$ be an integer, which for simplicity is assumed to divide $m$ (see Appendix~\ref{app:divisors} where we discuss how to get rid of this assumption).
Consider the conceptual matrix $M^q=\{m^q_{x,y}\}$ of size $\frac{m} {q} \times q$ where $m^q_{x,y} = p_{(x-1)\cdot q + y-1}$. An example is given in Figure~\ref{fig:Mq}.
For any integer $0\leq r< q$ the $r$'th column of $M^q$ corresponds to an \emph{offset pattern} $P_{q,r}=p_r p_{r+q} p_{r+2q}\dots p_{m-q+r}$.
Notice that some offset patterns might be equal.
Let $\Gamma_q = \{P_{q,r}\,|\,0\leq r<q \text{ and } '?'\notin P_{q,r}\}$ be the set of all the offset patterns that do not contain any wildcards.
Each offset pattern in  $\Gamma_q$ is given a unique id.
The set of unique ids is denoted by $\textit{ID}_q$.
We say that index $i$ in $P$ is \emph{covered} by $q$ if the column containing $p_i$ does not contain a wildcard, and so $P_{q,i\modulo q}\in \Gamma_q$.
The columns of $M^q$ define a \emph{column pattern} $P_q$ of length $q$, where the $j$'th character is the id of the $P_{q,j}$ column, or $'?'$ if $P_{q,j}\notin \Gamma_q$ (since $P_{q,j}$ contains wildcards).

We partition $T$ into $q$ \emph{offset texts}, where for every $0\leq r< q$ we define $T_{q,r}=t_r t_{r+q} t_{r+2q}\dots$.
Using the dictionary matching streaming (DMS) algorithm of Clifford et al.~\cite{CFPSS15B} we look for occurrences of offset patterns from $\Gamma_q$ in each of the offset texts.
We emphasize that we do \emph{not} only find occurrences of $P_{q,r}$ in $T_{q,r}$, since we cannot guarantee that the offset of $T$ synchronizes with an occurrence of $P$.
When the character $t_\alpha$ arrives, the algorithm passes $t_\alpha$ to the DMS algorithm for $T_{q,\alpha\modulo q}$.
We also create a streaming \emph{column text} $T_q$ whose characters correspond to the ids of offset patterns as follows.
If one of the offset patterns is found in $T_{q,\alpha\modulo q}$, then its id is the $\alpha$'th character in $T_q$. Otherwise, we use a dummy character for the $\alpha$'th character in $T_q$.

\paragraph{Full cover.}
Notice that an occurrence of $P$ in $T$ necessarily creates an occurrence of $P_q$ in $T_q$.
Such occurrences are found via the black box algorithm of Clifford et al.~\cite{CEPP11}.
However, an occurrence of $P_q$ in $T_q$ does not necessarily mean there was an occurrence of $P$ in $T$, since some characters in $P$ are not covered by $q$.
In order to avoid such false positives we run the process in parallel with several choices of $q$, while guaranteeing that each non wildcard character in $P$ is covered by at least one of those choices. Thus, if there is an occurrence of $P_q$ at location $i$ in $T_q$ for all the choices of $q$, then it must be that $P$ appears in $T$ at location $i$.  The choices of $q$ are given by the following lemma.

\begin{lemma}\label{lem:Qprimes}
There exists a set $Q$ of $O(\log d)$ prime numbers
such that any index of a non-wildcard character in $P$ is covered by at least one prime number $q\in Q$, and each number in $Q$ is at most $\tilde O(d)$.
\end{lemma}
\begin{proof}
The proof uses the probabilistic method: we show that the probability that the set $Q$ exists is strictly larger than 0. Since our proof is constructive it provides a randomized construction of $Q$.

It is well known that for a prime number $q$, every integer $0\leq z < q$ defines a congruence class which contains all integers $i$ such that $i \modulo q = z$.
For any two distinct natural numbers $x,y\in\mathbb N$, let $D_{x,y}$ be the set of prime numbers $q$ such that $x$ and $y$ are in the same congruence class modulo $q$ (i.e. $x\modulo q=y \modulo q$).
Notice that in the interpretation of the pattern columns in the conceptual matrix, if $q\in D_{x,y}$
then $p_x$ and $p_y$ are in the same column of the conceptual matrix $M^q$.
Recall that $W$ is the set of occurrences of wildcards in $P$.
Thus, if $0\le j<m$ is an index such that $j\notin W$ and if $w\in W$ such that $q\in D_{j,w}$, then $j$ is surely not covered by $q$.
By the Chinese remainder theorem, $|D_{j,w}|<\log m$ (otherwise for $\gamma=\prod_{q\in D_{j,w}}q>\prod_{q\in D_{j,w}}2\geq m$, and so $j\modulo \gamma =w \modulo \gamma$ implying that $j=w$).

For any $0\le j<m$  such that $j\notin W$, let $D_j=\bigcup_{w\in W}D_{j,w}$, so $|D_j|\leq \sum_{w\in W}|D_{j,w}|< |W|\log m=d\log m$.
If $2d\le \frac m {\log^2 m}$ then the proof is trivialized by choosing $Q$ to contain only the smallest prime number which is at least $m$.
If $2d>\frac m {\log^2 m}$, by Corollary 1 in~\cite{RSL62},  then there are at least $2d\log  m$ prime numbers whose value  are upper bounded by $2d\log^2 m$.
Let $\hat{Q}$ be the set of those prime numbers.
For a random $q\in \hat{Q}$, the probability that a specific non-wildcard pattern index $j$ is not covered by $q$ is at most $\frac{|D_j|}{|\hat{Q}|}\leq\frac{d\log m}{2d\log m}=\frac{1}{2}$.
Let $Q$ be a set of  $2\log m$ randomly chosen prime numbers from $\hat{Q}$.
The probability that a specific non-wildcard pattern index $j$ is not covered by any of the prime numbers in $Q$ is less than $\frac{1}{2^{2\log m}}\leq \frac{1}{m^2}$. Thus, the probability that there exists a non-wildcard pattern index $j$  which is not covered by any of the prime numbers in $Q$ is less than $\frac{m-d}{m^2}\leq \frac{1}{m}$. Therefore, there must exist a set $Q$ that covers all of the indices of non-wildcard characters from $P$.
\end{proof}

From a space usage perspective, we need the size of $|\Gamma_q|$ to be small, since this directly affects the space usage of the DMS algorithm which uses $\tilde{O}(k)$ space, where $k$ is the number of patterns in the dictionary. In our case $k = |\Gamma_q|$. In order to bound the size of $\Gamma_q$ we use the following lemma.

\begin{lemma}\label{lem:GammaBound}
If $\pi_P \leq \tau$ then for any $q\in \mathbb N$ we have $|\Gamma_q|=O(\tau)$.
\end{lemma}

\begin{proof}
Since $\pi_P \leq \tau$, there exists a string $S=s_0 \dots s_{2m-2}$ with no wildcards that contains $\Omega(\frac{m}{\tau})$ occurrences of $P$. Using the string $S$ we show that $|\Gamma_q|=O(\tau)$.

For each id in $\textit{ID}_q$ we pick an index of a representative column in $M_q$ that has this id, and denote this set by $R_q$.
Let $r_1$ be the minimum index in $R_q$.
For every index $0\leq i<m$ let $S_i=s_i\dots s_{i+m-1}$ (see Figure~\ref{fig:si-example}).
For every $0\leq r<q$ let $S_{i,q,r}=s_{i+r}s_{i+r+q}\dots s_{i+m-q+r}$, and so for any integer $0\leq \Delta< q-r$ we have $S_{i,q,r+\Delta}=S_{i+\Delta,q,r}$.
Notice that if $S_i$ matches $P$ then $P_{q,r}=S_{i,q,r}$ for each $r\in R_q$.

\begin{figure}
\centering
\includegraphics[width=\textwidth]{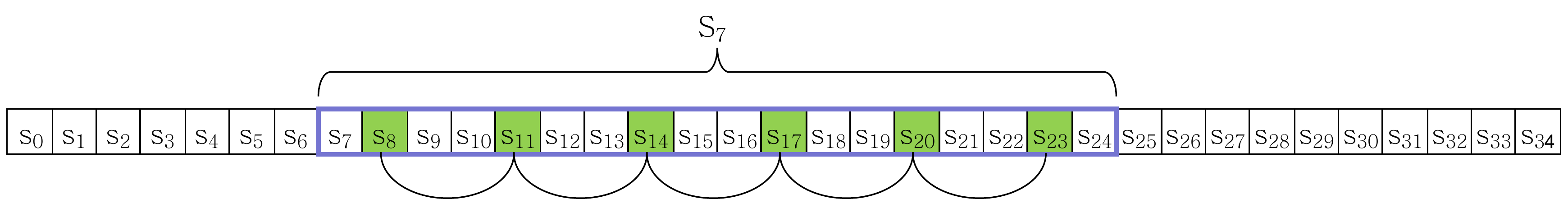}
\caption{For string $S=s_0\dots s_{34}$  for pattern of length $m=18$, $S_7$ is marked by the blue rectangle and the green indices are the characters of $S_{7,3,1}$. Notice that $S_{7,3,1}=S_{8,3,0}$.}

\label{fig:si-example}
\end{figure}

Let $i$ be an index of an occurrence of $P$ in $S$.
For any distinct $r,r'\in R_q$, it must be that $S_{i,q,r} = P_{q,r} \neq P_{q,r'} = S_{i,q,r'}$.
In particular, for any $r\in R_q$ such that $r > r_1$, we have $P_{q,r_1}=S_{i,q,r_1}\neq S_{i,q,r}=S_{i+r-r_1,q,r_1}$.
This implies that $i+r-r_1$ cannot be an occurrence of $P$.
Hence, every occurrence of $P$ in $S$ eliminates $|R_q|-1$ locations in $S$ from being an occurrence of $P$.
We now show that the sets of eliminated locations defined by distinct occurrences are disjoint.
Assume without loss of generality that $S$ contains at least two occurrences.
Let  $i_1$ and $i_2$ be two distinct occurrences of $P$ in $S$, and assume by contradiction that an index $j$ is eliminated by both of these occurrences.
Since $s_{i_1}\dots s_{i_1+m-1}$ matches $P$, we have that $S_{i_1,q,j-i_1}=P_{j-i_1}$ and $j-i_1\in R_q$.
Similarly, we have that $S_{i_2,q,j-i_2}=P_{j-i_2}$ and $j-i_2\in R_q$.
Being that $S_{i_1,q,j-i_1}=S_{i_2,q,j-i_2}$ we have that $P_{j-i_2}=P_{j-i_1}$, contradicting the definition of $R_q$.

Therefore, the maximum number of occurrences of $P$ in $S$ is at most $\frac{|S|}{|R_q|}=\frac{2m-1}{|R_q|}$.
Since $S$ contains at least $\frac{m}{\tau}$  instances of $P$, it must be that $\frac{m}{\tau}\leq \frac{2m-1}{|R_q|}$ which implies that $|\Gamma_q|=|R_q|\leq 2\tau=O(\tau)$.
\end{proof}

\paragraph{Complexities.}
For a single $q\in Q$, the algorithm creates $q = \tilde{O}(d)$ offset patterns and texts.
For each such offset text the algorithm applies an instance of the DMS algorithm with a dictionary of $O(\tau)$ strings (by Lemma~\ref{lem:GammaBound}).
Since each instance of the DMS algorithm uses $\tilde{O}(\tau)$ words of space~\cite{CFPSS15B}, the total space usage for all instances of the DMS algorithm is $\tilde{O}(d\tau)$ words. Moreover, the time per character in each DMS algorithm is $\tilde{O}(1)$ time, and each time a character appears we inject it into only one of the DMS algorithms (for this specific $q$).
In addition, the algorithm uses an instance of the black box algorithm for $T_q$, with a pattern of length $q$. This uses another $O(q)= \tilde{O}(d)$ space and another $\tilde{O}(1)$ time per character~\cite{CEPP11}. Thus the total space usage due to one element in $Q$ is $\tilde{O}(d\tau)$ words. Since $|Q|=O(\log d)$ the total space usage for all elements in $Q$ is $\tilde{O}(d\tau)$ words, and the total time per arriving character is $\tilde{O}(1)$. Thus we have proven the following.

\begin{theorem} \label{thm:smallWPAlgorithm}

For any $\tau\ge 1$, there exists a randomized Monte Carlo algorithm for PMDW  on  patterns $P$ with $\pi_P<\tau$ in the streaming model,	which succeeds with probability $1-1/poly(n)$, uses $\tilde O(d\tau)$ words of space and spends $\tilde O(1)$ time per arriving text character.
\end{theorem}

\subsection{Proof of Theorem~\ref{thm:tradeoff}}\label{sec:proof-tradeoff}
In this section we combine the algorithm of Theorem~\ref{thm:main} with the algorithm of Theorem~\ref{thm:smallWPAlgorithm} and introduce an algorithm for patterns with general wildcard-period length, thereby proving Theorem~\ref{thm:tradeoff}.

Prior to Section~\ref{sec:smallPiAlgo} we presented an almost accurate description of the algorithm. The only two parts of the description that require elaboration are regarding how to insert occurrences of $P^*$ into the appropriate candidate-fingerprint-queue efficiently, and how to schedule validations of candidates so that the amortized cost is low. We first focus on how to insert candidates and later we discuss the scheduling.

\paragraph{Direct insertion of candidates.} The challenge with inserting occurrences of $P^*$ into $Q_{I^*}$ is that the candidate-fingerprint-queue data structure uses the \SI{} of candidates, and so the straightforward ways for providing this information together with the new candidates (which are occurrences of $P^*$) cost either too much time or too much space.
In order to meet our desired complexities, we first investigate the purposes of different parts of \SI{}.

The \SI{} for a candidate $c$ in $\candidates{I=[i,j]}{\alpha}$ consists of the candidate fingerprint, $\phi(t_0\dots t_{c-1})$, and the entrance fingerprint, $\phi(t_c\dots t_{c+i-1})$. The \SI{} has two purposes. The first is to validate a candidate after a $\Dequeue()$ operation, in which case the algorithm makes use of both parts of the \SI{} in order to compute $\phi(t_{c+i}\dots t_{c+j})$ by combining the \SI{} with the text fingerprint. The second purpose is to compute the next entrance fingerprints of candidates in order to
distinguish between candidates that are stored as part of an arithmetic progression and candidates that are not. The entrance fingerprint is obtained, via the sliding property, from the candidate fingerprint in the \SI{} and the current text fingerprint.

Notice that in order to validate $c$ the algorithm only needs the fingerprint of $\phi(t_0\dots t_{c-i+1})$.
Also notice that entrance prefixes are only used for candidates that are at some point part of a stored arithmetic progression. Thus, for a specially chosen subset of strings $\Psi\subseteq \Sigma^{|P^*|}$ we precompute all of the fingerprints of strings in $\Psi$. The set $\Psi$ is chosen so that for any occurrence of $P^*$ that is injected as a candidate $c$ where $c$ is at some point part of a stored arithmetic progression, the occurrence of $P^*$ at location $c$ is in $\Psi$. We use the DMS algorithm~\cite{CFPSS15B} to locate strings from $\Psi$ in the text, and whenever such a string appears, we compute the \SI{} for the corresponding candidate in constant time from the stored fingerprint and the current text fingerprint. We emphasize that not all of the candidates that correspond to strings in $\Psi$ need to necessarily at some point be a part of an arithmetic progression. However, in order to reduce the space usage, we require that $\Psi$ is not too large, and in particular $|\Psi| = O(d+\log m)$. For a candidate $c$ that does not correspond to a string in $\Psi$, instead of maintain the \SI{} of $c$, we explicitly maintain the fingerprint of $\phi(t_0\dots t_{c-i+1})$ where $c\in \candidates{I=[i,j]}{\alpha}$. Notice that whenever such a candidate enters a new text interval, the text fingerprint at that time is exactly the information which we need to store.

\paragraph{Creating $\Psi$.}
Consider all pattern intervals $I=[i,j]\in \mathcal{I}$ with $i\geq i^*$.
Notice that there are at most $O(d + \log m)$ such pattern intervals.
For each such interval $I$, let $\psi_I$ be the prefix of $u_I$ of length $|P^*|$.
Since, by Lemma~\ref{lem:u_I}, a candidate $c \in \candidatesUi{I}{\alpha}$ implies an occurrence of $u_I$ at location $c$, then $\psi_I$ also appears at location $c$. Thus, we define $\Psi$ to be the set containing $\psi_I$ for all such pattern intervals $I$. Since any candidate in an arithmetic progression at time $\alpha$ must be in $\candidatesUi{I}{\alpha}$ for some interval $I$, it is guaranteed that when $c$ corresponded to an occurrence of $P^*$, that occurrence must have been $\psi_I$, and so $\Psi$ has the required properties.

\paragraph{Scheduling validations.}
Since the only bound we have proven on the number of pattern intervals $I=[i,j]\in \mathcal{I}$ with $i\geq i^*$ is $O(d+\log m)$, if each time a new text character arrives we perform a $\Dequeue()$ operation for each one of the pattern intervals, then the time cost can be as large as $O(d+\log m)$ which is too much. The solution for reducing this time cost is to only perform a $\Dequeue()$ operation on $Q_I$ when a candidate $c$ actually leaves $\ti{I}{\alpha}$ and needs to be validated. This is implemented by maintaining a priority queue on top of the pattern intervals, where the keys that are used are the next time a candidate exits the corresponding text interval. Each time a candidate leaves a text interval, the key for the queue of that interval is updated to the time the next candidate leaves (if such a candidate exists). When a candidate entering a text interval is the only candidate of that text interval, then the key for the queue of this text interval is also updated.

\paragraph{Complexities.}
Recall that $I^*=[i^*,j^*]$ is a pattern interval such that $i^*=|P^*|$, and that each time the algorithm finds an occurrence of $P^*$, the corresponding candidate is inserted into $Q_{I^*}$. Let $P'$ be the prefix of $P$ of length $j^*+1$. By maximality of $P^*$, it must be that $\pi_{P'}>d^\delta$.
We partition the time usage of the algorithm into three parts. The first is the amount of time spent on finding occurrences of $P^*$ using the algorithm of Theorem~\ref{thm:smallWPAlgorithm}, which is $\tilde O(1)$. The second is the amount of time spent performing $\Enqueue()$ and $\Dequeue()$ operations on $Q_{I^*}$, which is also $\tilde O(1)$ since we perform $O(1)$ operations on this queue per each arriving character. The third is the amount of time spent on $\Enqueue()$ and $\Dequeue()$ operations on $Q_I$ for $I=[i,j]$ with $i>j^*$. These operations only apply to candidates that are occurrences of $P'$. For this part we use amortized analysis.

By definition of wildcard-period length, for any string $S$ of size $2|P'|-1$, we have $d^\delta < \pi_{P'} \le \ceil{\frac{|P'|}{occ(S,P')}}$. Being that $occ(S,P') \le |P'|$, we have $d^\delta < \frac{2|P'|}{occ(S,P')}$. Notice that for a text $T$ of size $n\ge |P'|$, we must have $occ(T,P')<\frac{2n}{d^\delta}$. This is because otherwise, if $n \ge 2|P'|-1$ then there exists a substring of $n$ of length $2|P'|-1$ with at least $\frac{2|P'|}{d^\delta}$ occurrences of $P'$, and if $n < 2|P'|-1$ then we can pad $T$ to create such a string. In both cases we contradict $d^\delta < \frac{2|P'|}{occ(S,P')}$ for any string $S$ of length $2|P'|-1$.

The total amount of time spent on each occurrence of $P'$ is $\tilde O(d)$, and so the total cost for processing $T$ on candidates that are also occurrences of $P'$ is at most $\tilde O(occ(T,P')\cdot d) = \tilde O(\frac{2n}{d^\delta}d)=O(n\cdot d^{1-\delta})$. Thus, the amortized cost per character is $\tilde O(d^{1-\delta})$.

For the space complexity, the most expensive part is the use of the algorithm of Theorem~\ref{thm:smallWPAlgorithm} which takes $O(d\cdot d^\delta)=O(d^{1+\delta})$ words of space. This completes the proof of Theorem~\ref{thm:tradeoff}.

\appendix

\bibliography{StreamingWildcards}

\section{Missing Details}\label{sec:tradeoffAppendix}
\subsection{Dealing with \texorpdfstring{$q\nmid m$}{ q nmid m}}\label{app:divisors}
If $q\nmid m$, then the strings in $\Gamma_q$ have two possible lengths; either $\floor{\frac{m}{q}}$ or $\ceil{\frac{m}{q}}$. This implies that one string in $\Gamma_q$ could be a proper suffix of another string in $\Gamma_q$. So if the longer one appears in an offset text, then both ids need to be given to $T_q$ - a situation in which it is not clear what to do.
So to avoid such scenarios, for each $q\in Q$ we run the algorithm twice, in parallel, where one instance uses the DMS algorithm for one length while the other instance uses the DMS algorithm on the other length.
This creates two instances of $P_q$ and $T_q$, one for each length of columns under consideration. Notice that in order for the algorithm to work, when considering one specific length, all of the columns that correspond to the other length are treated as a $'?'$ in the appropriate instance of $P_q$.

\end{document}